\documentclass{llncs}

\usepackage[T1]{fontenc}
\usepackage[utf8]{inputenc}

\usepackage{amssymb, amsmath,graphicx,graphics,enumerate,tkz-graph}
\usepackage{comment}
\usepackage{microtype}
\usepackage[hypertexnames=false]{hyperref}

\title{Hereditary Graph Classes: When the Complexities of Colouring and Clique Cover Coincide\thanks{This paper received support from EPSRC (EP/K025090/1) and the Leverhulme Trust (RPG-2016-258).}}

\author{Alexandre Blanch\'e\inst{1} \and Konrad K. Dabrowski\inst{2} \and Matthew Johnson\inst{2} \and Dani\"el~Paulusma\inst{2}}

\institute{\'Ecole normale sup\'erieure de Rennes,\\
D\'epartement Informatique et T\'el\'ecommunications,\\
Campus de Ker Lann, Avenue Robert Schuman, 35170 Bruz, France\\
\texttt{alexandre.blanche@ens-rennes.fr}
\and
School of Engineering and Computing Sciences, Durham University,\\
Science Laboratories, South Road, Durham DH1 3LE, United Kingdom\\
\texttt{\{konrad.dabrowski,matthew.johnson2,daniel.paulusma\}@durham.ac.uk}}

\newcommand{\ssi}{\subseteq_i}

\newcommand{\NP}{{\sf NP}}

\DeclareMathOperator{\cw}{cw}

\newcounter{ctrclaim}[theorem]

\newcounter{ctrcase}[theorem]
\newcounter{ctrsubcase}[ctrcase]

\renewcommand{\thectrsubcase}{\thectrcase\alph{ctrsubcase}}

\newcommand\displaycase[1]{{\bf #1}}
\newcommand\faketheorem[1]{\displaycase{#1}}
\newcommand{\qedllncs}{\qed}

\newcommand{\clm}[1]{\setcounter{ctrcase}{0}\medskip\phantomsection\refstepcounter{ctrclaim}\noindent\displaycase{Claim \thectrclaim. }{\em #1}\\}
\newcommand{\clmnonewline}[1]{\setcounter{ctrcase}{0}\medskip\phantomsection\refstepcounter{ctrclaim}\noindent\displaycase{Claim \thectrclaim. }{\em #1}}
\newcommand{\shortclm}[1]{\setcounter{ctrcase}{0}\phantomsection\refstepcounter{ctrclaim}\noindent\displaycase{Claim \thectrclaim. }{\em #1}}
\newcommand{\thmcase}[1]{\medskip\phantomsection\refstepcounter{ctrcase}\noindent\displaycase{Case \thectrcase. }{\em #1}\\}
\newcommand{\thmsubcase}[1]{\medskip\phantomsection\refstepcounter{ctrsubcase}\noindent\displaycase{Case \thectrsubcase. }{\em #1}\\}

\begin{document}
\maketitle
\begin{abstract}
A graph is $(H_1,H_2)$-free for a pair of graphs $H_1,H_2$ if it contains no induced subgraph isomorphic to~$H_1$ or~$H_2$.
In 2001, Kr\'al', Kratochv\'{\i}l, Tuza, and Woeginger initiated a study into the complexity of {\sc Colouring} for $(H_1,H_2)$-free graphs. 
Since then, others have tried to complete their study, but many cases remain open.
We focus on those $(H_1,H_2)$-free graphs where~$H_2$ is~$\overline{H_1}$, the complement of~$H_1$.  
As these classes are closed under complementation, the computational complexities of {\sc Colouring} and {\sc Clique Cover} coincide. 
By combining new and known results, we are able to classify the complexity of {\sc Colouring} and {\sc Clique Cover} for $(H,\overline{H})$-free graphs 
for all cases
except when $H=sP_1+\nobreak P_3$ for $s\geq 3$ or $H=sP_1+\nobreak P_4$ for $s\geq 2$.  
We also classify the complexity of {\sc Colouring} on graph classes characterized by forbidding a finite number of self-complementary induced subgraphs, and we initiate a study of $k$-{\sc Colouring} for $(P_r,\overline{P_r})$-free graphs.
\end{abstract}

\section{Introduction}\label{s-intro}

A colouring of a graph is an assignment of labels, called colours, to its vertices in such a way that no two adjacent vertices have the same label. 
The corresponding decision problem, {\sc Colouring}, which is that of deciding whether a given graph can be coloured with at most~$k$ colours for some given positive integer~$k$, is a central problem in discrete optimization. 
Its complementary problem {\sc Clique Cover} is that of deciding whether the vertices of a graph can be covered with at most~$k$ cliques.
As {\sc Colouring} and {\sc Clique Cover} are \NP-complete even for $k=3$~\cite{Lo73}, it is natural to restrict their input to some special graph class. 
A classic result in this area is  due to Gr\"otschel, Lov\'asz, and Schrijver~\cite{GLS84}, who showed that {\sc Colouring} is polynomial-time solvable for perfect graphs.
However, for both problems, finding the exact borderline between tractable and computationally hard graph classes is still open.

A graph class is {\it hereditary} if it can be characterized by some set of forbidden induced subgraphs, or equivalently, if it is closed under vertex deletion.
The aforementioned class of perfect graphs is an example of such a class, as a graph is perfect if and only if it contains no induced odd holes and no induced odd antiholes~\cite{CRST06}.
For the case where exactly one induced subgraph is forbidden, Kr\'al', Kratochv\'{\i}l, Tuza, and Woeginger~\cite{KKTW01} were able to prove a complete dichotomy; see also \figurename~\ref{fig:P1+P3-P4} (we write $G'\ssi G$ to denote that the graph~$G'$ is an {\em induced} subgraph of the graph~$G$).

\begin{theorem}[\cite{KKTW01}]\label{t-col-H-free}
Let~$H$ be a graph.
The {\sc Colouring} problem on $H$-free graphs is polynomial-time solvable if $H \ssi P_1+\nobreak P_3$ or $H \ssi P_4$ and it is \NP-complete otherwise.
\end{theorem}

\begin{figure}[h]
\begin{center}
\begin{tabular}{cc}
\begin{minipage}{0.33\textwidth}
\centering
\scalebox{0.75}{
{\begin{tikzpicture}[scale=1]
\GraphInit[vstyle=Simple]
\SetVertexSimple[MinSize=6pt]
\Vertex[x=0,y=1]{a}
\Vertex[x=1,y=1]{b}
\Vertex[x=2,y=1]{c}
\Vertex[x=3,y=1]{d}
\Edges(b,c,d)
\end{tikzpicture}}}
\end{minipage}
&
\begin{minipage}{0.33\textwidth}
\centering
\scalebox{0.75}{
{\begin{tikzpicture}[scale=1]
\GraphInit[vstyle=Simple]
\SetVertexSimple[MinSize=6pt]
\Vertex[x=0,y=1]{a}
\Vertex[x=1,y=1]{b}
\Vertex[x=2,y=1]{c}
\Vertex[x=3,y=1]{d}
\Edges(a,b,c,d)
\end{tikzpicture}}}
\end{minipage}\\
\\
$P_1+\nobreak P_3$ & $P_4$\\
\end{tabular}
\end{center}
\caption{\label{fig:P1+P3-P4}The graphs~$H$ such that {\sc Colouring} can be solved in polynomial time on $H$-free graphs.}
\end{figure}
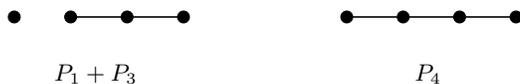

\noindent
We study the complexity of {\sc Colouring} and {\sc Clique Cover} for $(H_1,H_2)$-free graphs, that is, when two induced subgraphs~$H_1$ and~$H_2$ are both forbidden.
Both polynomial-time and \NP-completeness results are known for $(H_1,H_2)$-free graphs, as shown by many teams of researchers (see, for example,~\cite{BGPS12a,CH,DDP15,DGP14,HH,HL,KKTW01,LM15,Ma13,Ma,ML17,Sc05}), but the complexity classification is far from complete:
even if we forbid two graphs~$H_1$ and~$H_2$ of up to four vertices, there are still three open cases left, namely when $(H_1,H_2)\in \{(K_{1,3},4P_1),\allowbreak (K_{1,3},2P_1+\nobreak P_2),\allowbreak (C_4,4P_1)\}$; see~\cite{LM15} (the graph~$K_{1,3}$ is the claw and~$C_4$ is the $4$-vertex cycle).
We refer to~\cite[Theorem 21]{GJPS} for a summary and to the recent paper~\cite{ML17}, in which the number of missing cases when~$H_1$ and~$H_2$ are both connected graphs on at most five vertices was reduced from ten to eight.

To narrow the complexity gap between hard and easy cases and to increase our understanding of it, we initiate a systematic study into $(H_1,H_2)$-free graph classes that are {\em closed under complementation}: 
if a graph~$G$ belongs to the class, then so does its {\it complement}~$\overline{G}$, which is the graph with vertex set~$V(G)$ and an edge between two distinct vertices if and only if these two vertices are not adjacent in~$G$.
For such graph classes the complexities of {\sc Colouring} and {\sc Clique Cover} coincide, so
we only need to consider the {\sc Colouring} problem. 
A graph~$H$ is {\em self-complementary} if~$H=\overline{H}$.
We observe that a class of $(H_1,H_2)$-free graphs is closed under complementation if either~$H_1$ and~$H_2$ are both self-complementary, or $H_2=\overline{H_1}$.
For the first case we prove the following more general classification in Section~\ref{s-self}.

\begin{sloppypar}
\begin{theorem}\label{t-self}
Let $H_1,\ldots,H_k$ be self-complementary graphs.
Then {\sc Colouring} is polynomial-time solvable for $(H_1,\ldots,H_k)$-free graphs if $H_i \ssi P_4$ for some~$i$ and it is \NP-complete otherwise.
\end{theorem}
\end{sloppypar}

\begin{sloppypar}
Hence, we may now focus on the case when $H_2=\overline{H_1}$ and~$H_1$ is not self-complementary.
In this case few results are known.
First notice that for any integer $s\geq 1$, the class of $(sP_1,\overline{sP_1})$-free graphs consists of graphs with no large independent set and no large clique.
The number of vertices in such graphs is bounded by a constant by Ramsey's Theorem.
Dabrowski et al.~\cite{DGP14} researched the effect on the complexity of changing the forbidden subgraph~$sP_1$ by adding an extra edge, and proved that {\sc Colouring} is polynomial-time solvable for $(sP_1+P_2,\overline{sP_1+P_2})$-free graphs. 
A result of Malyshev~\cite{Ma13} implies that {\sc Colouring} is polynomial-time solvable for $(K_{1,3},\overline{K_{1,3}})$-free graphs.
Ho\`ang and Lazzarato~\cite{HL} proved the same for $(P_5,\overline{P_5})$-free graphs.
By combining results of~\cite{KR03b} and~\cite{OS06}, {\sc Colouring} is polynomial-time solvable on graph classes of bounded clique-width.
It is known that $(P_1+\nobreak P_4,\allowbreak \overline{P_1+P_4})$-free graphs~\cite{BLM04} and $(2P_1+\nobreak P_3,\allowbreak \overline{2P_1+P_3})$-free graphs~\cite{BDJLPZ} have bounded clique-width, so {\sc Colouring} is polynomial-time solvable for these two graph classes. 
In light of these results, it is a natural research question to find out to what extent we may add edges to~$sP_1$ so that {\sc Colouring} 
remains polynomial-time solvable on $(H,\overline{H})$-free graphs for the resulting~graph~$H$.
\end{sloppypar}

To narrow down the large number of open cases $(H,\overline{H})$, we will use the known results that {\sc Colouring} is \NP-complete for graphs of girth at least~$p$ for any $p\geq\nobreak 3$ \cite{EHK98,KKTW01,KL07} (the girth of a graph is the length of a shortest induced cycle in it) and for $(K_{1,3},K_4,\overline{2P_1+P_2})$-free graphs~\cite{KKTW01}.
Even after using these results there are still many open cases left, and to get a handle on them we need another \NP-hardness result for a very restricted graph class.
In Section~\ref{sec:hard} we will prove such a result, namely, we show that {\sc Colouring} is \NP-complete even for $(P_1+\nobreak 2P_2,\allowbreak \overline{P_1+2P_2},\allowbreak 2P_3,\allowbreak \overline{2P_3},\allowbreak P_6,\overline{P_6})$-free graphs (see \figurename~\ref{fig:npc}).
We will show that, by combining this result with those above, only the following open cases are left: $H=P_2+\nobreak P_3$ (see \figurename~\ref{fig:H-co-H-unknown-colouring}), $H=(s+1)P_1+P_3$ or $H=sP_1+P_4$ for $s\geq 2$.
In Section~\ref{sec:easy} we give a polynomial-time algorithm for {\sc Colouring} restricted to the class of $(P_2+\nobreak P_3,\allowbreak \overline{P_2+P_3})$-free graphs (note that this class has unbounded clique-width, as it contains the class of 
split graphs, which have unbounded clique-width~\cite{MR99}).

\begin{figure}
\begin{center}
\begin{tabular}{cccccc}
\begin{minipage}{0.15\textwidth}
\centering
\scalebox{0.75}{
{\begin{tikzpicture}[scale=1,rotate=90]
\GraphInit[vstyle=Simple]
\SetVertexSimple[MinSize=6pt]
\Vertices{circle}{a,b,c,d,e}
\Edges(b,c)
\Edges(d,e)
\end{tikzpicture}}}
\end{minipage}
&
\begin{minipage}{0.15\textwidth}
\centering
\scalebox{0.75}{
{\begin{tikzpicture}[scale=1,rotate=90]
\GraphInit[vstyle=Simple]
\SetVertexSimple[MinSize=6pt]
\Vertices{circle}{a,b,c,d,e}
\Edges(b,c,d,e,b)
\Edges(b,a,c)
\Edges(d,a,e)
\end{tikzpicture}}}
\end{minipage}
&
\begin{minipage}{0.15\textwidth}
\centering
\scalebox{0.75}{
{\begin{tikzpicture}[scale=1]
\GraphInit[vstyle=Simple]
\SetVertexSimple[MinSize=6pt]
\Vertices{circle}{a,b,c,d,e,f}
\Edges(c,d,e)
\Edges(f,a,b)
\end{tikzpicture}}}
\end{minipage}
&
\begin{minipage}{0.15\textwidth}
\centering
\scalebox{0.75}{
{\begin{tikzpicture}[scale=1]
\GraphInit[vstyle=Simple]
\SetVertexSimple[MinSize=6pt]
\Vertices{circle}{a,b,c,d,e,f}
\Edges(a,d,e,f,a,e)
\Edges(f,d)
\Edges(b,c)
\Edges(a,b,f)
\Edges(d,c,e)
\end{tikzpicture}}}
\end{minipage}
&
\begin{minipage}{0.15\textwidth}
\centering
\scalebox{0.75}{
{\begin{tikzpicture}[scale=1]
\GraphInit[vstyle=Simple]
\SetVertexSimple[MinSize=6pt]
\Vertices{circle}{a,b,c,d,e,f}
\Edges(c,d,e,f,a,b)
\end{tikzpicture}}}
\end{minipage}
&
\begin{minipage}{0.15\textwidth}
\centering
\scalebox{0.75}{
{\begin{tikzpicture}[scale=1]
\GraphInit[vstyle=Simple]
\SetVertexSimple[MinSize=6pt]
\Vertices{circle}{a,b,c,d,e,f}
\Edges(f,a,b,f)
\Edges(c,d,e,c)
\Edges(a,d)
\Edges(b,c,f,e)
\end{tikzpicture}}}
\end{minipage}\\
\\
$P_1+\nobreak 2P_2$ &
$\overline{P_1+2P_2}$&
$2P_3$&
$\overline{2P_3}$&
$P_6$&
$\overline{P_6}$
\end{tabular}
\end{center}
\caption{\label{fig:npc}The six graphs corresponding to our result that {\sc Colouring} is \NP-complete for $(P_1+\nobreak 2P_2,\allowbreak \overline{P_1+2P_2},\allowbreak 2P_3,\allowbreak \overline{2P_3},\allowbreak P_6,\overline{P_6})$-free graphs.}
\end{figure}
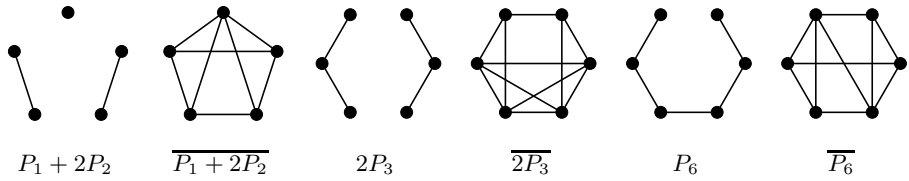

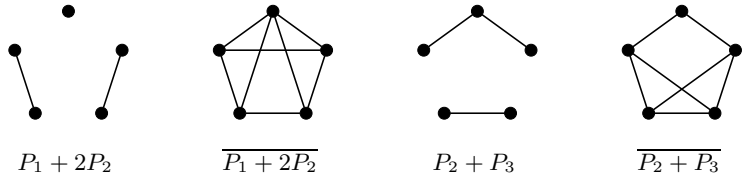
\begin{figure}
\begin{center}
\begin{tabular}{cccccc}
\begin{minipage}{0.20\textwidth}
\centering
\scalebox{0.75}{
{\begin{tikzpicture}[scale=1,rotate=90]
\GraphInit[vstyle=Simple]
\SetVertexSimple[MinSize=6pt]
\Vertices{circle}{a,b,c,d,e}
\Edges(b,c)
\Edges(d,e)
\end{tikzpicture}}}
\end{minipage}
&
\begin{minipage}{0.20\textwidth}
\centering
\scalebox{0.75}{
{\begin{tikzpicture}[scale=1,rotate=90]
\GraphInit[vstyle=Simple]
\SetVertexSimple[MinSize=6pt]
\Vertices{circle}{a,b,c,d,e}
\Edges(b,c,d,e,b)
\Edges(b,a,c)
\Edges(d,a,e)
\end{tikzpicture}}}
\end{minipage}
&
\begin{minipage}{0.20\textwidth}
\centering
\scalebox{0.75}{
{\begin{tikzpicture}[scale=1,rotate=90]
\GraphInit[vstyle=Simple]
\SetVertexSimple[MinSize=6pt]
\Vertices{circle}{a,b,c,d,e}
\Edges(e,a,b)
\Edges(c,d)
\end{tikzpicture}}}
\end{minipage}
&
\begin{minipage}{0.20\textwidth}
\centering
\scalebox{0.75}{
{\begin{tikzpicture}[scale=1,rotate=90]
\GraphInit[vstyle=Simple]
\SetVertexSimple[MinSize=6pt]
\Vertices{circle}{a,b,c,d,e}
\Edges(e,a,b)
\Edges(b,c,e,d,b)
\Edges(c,d)
\end{tikzpicture}}}
\end{minipage}\\
\\
$P_1+\nobreak 2P_2$ &
$\overline{P_1+\nobreak 2P_2}$ &
$P_2+\nobreak P_3$ &
$\overline{P_2+\nobreak P_3}$
\end{tabular}
\end{center}
\caption{\label{fig:H-co-H-unknown-colouring}The graphs~$H$ on five vertices for which the complexity of {\sc Colouring} for $(H,\overline{H})$-free graphs was unknown. 
We show \NP-completeness if $H \in \{P_1+\nobreak 2P_2,\overline{P_1+\nobreak 2P_2}\}$ and polynomial-time solvability if $H \in \{P_2+\nobreak P_3,\overline{P_2+P_3}\}$.}
\end{figure}

In Section~\ref{s-final} we show that our new hardness result and new polynomial-time algorithm combine together with the aforementioned known results to give us the main theorem of our paper.

\begin{sloppypar}
\begin{theorem}\label{t-mainmain}
Let~$H$ be a graph with $H,\overline{H} \notin \{(s+\nobreak 1)P_1+\nobreak P_3, sP_1+\nobreak P_4\; |\; s\geq 2\}$.
Then {\sc Colouring} is polynomial-time solvable for $(H,\overline{H})$-free graphs if
$H$ or $\overline{H}$ is an induced subgraph of $K_{1,3}$,
$P_1+\nobreak P_4$,
$2P_1+\nobreak P_3$,
$P_2+\nobreak P_3$,
$P_5$, or $sP_1+\nobreak P_2$ for some $s\geq 0$
and it is \NP-complete otherwise.
\end{theorem}
\end{sloppypar}

Theorem~\ref{t-mainmain} tells us is that we may only add a few edges to the graph $sP_1$ if {\sc Colouring} stays polynomial-time solvable on $(H,\overline{H})$-free graphs for the resulting graph~$H$ (even in the two missing cases where $H=(s+1)P_1+\nobreak P_3$ or $H=sP_1+\nobreak P_4$ for $s \geq 2$, the graph~$H$ is only allowed to contain at most three edges).
From a more general perspective, our new hardness result has significantly narrowed the classification for $(H_1,H_2)$-free graphs.

Our immediate goal is to complete the complexity classification of {\sc Colouring} for $(H,\overline{H})$-free graphs, thus to solve the cases when $H=(s+1)P_1+\nobreak P_3$ or $H=sP_1+\nobreak P_4$ for $s\geq 2$.
We note that the class of $(3P_1+\nobreak P_2,K_4)$-free graphs, and thus the class of $(3P_1+\nobreak P_3,\allowbreak \overline{3P_1+P_3}$)-free graphs, has unbounded clique-width~\cite{DGP14}.
We emphasize that our long-term goal is to increase our understanding of the computational complexity of {\sc Colouring} for hereditary graph classes.
Another natural question is whether $k$-{\sc Colouring} (the variant of {\sc Colouring} where the number of colours is fixed) is polynomial-time solvable for $(H,\overline{H})$-free graphs when $H=(s+1)P_1+\nobreak P_3$ or $H=sP_1+\nobreak P_4$ for $s\geq 2$.
This is indeed the case, as Couturier et al.~\cite{CGKP15} extended the result of~\cite{HKLSS10} on $k$-{\sc Colouring} for $P_5$-free graphs by proving that for every pair of integers $k,s\geq 1$, {\sc $k$-Colouring} is polynomial-time solvable even for $(sP_1+\nobreak P_5)$-free graphs.
However, for other classes of $(H,\overline{H})$-free graphs, we show that $k$-{\sc Colouring} turns out to be \NP-hard by using a construction of Huang~\cite{Hu16}, that is, we show the following results in Section~\ref{sec:k-col}.

\begin{theorem}\label{thm:4-col-overP8}
$4$-{\sc Colouring} is \NP-complete for $(P_7,\overline{P_8})$-free graphs, and thus for $(P_8,\overline{P_8})$-free graphs.
\end{theorem}

\begin{theorem}\label{thm:5-col-overP8}
$5$-{\sc Colouring} is \NP-complete for $(P_6,\overline{P_1+P_6})$-free graphs, and thus for $(P_1+\nobreak P_6,\overline{P_1+P_6})$-free graphs.
\end{theorem}

As {\sc Colouring} is polynomial-time solvable for $(P_5,\overline{P_5})$-free graphs~\cite{HL}, it would be interesting to
solve the following two open problems (see also Table~\ref{tbl:pt-co-pt-col}):
\begin{itemize}
\item is there an integer $k$ such that $k$-{\sc Colouring} for $(P_6,\overline{P_6})$-free graphs is \NP-complete?
\item is there an integer~$k$ such that $k$-{\sc Colouring} for $(P_7,\overline{P_7})$-free graphs is \NP-complete?
\end{itemize}
As $3$-{\sc Colouring} is polynomial-time solvable for $P_7$-free graphs~\cite{BCMSZ}, $k$ must be at least~$4$.
Similar to summaries for $k$-{\sc Colouring} for $P_t$-free graphs (see, for example,~\cite{GJPS}) and $(C_s,P_t)$-free graphs~\cite{HH}, we can survey the known results and the missing cases of $k$-{\sc Colouring} for $(P_t,\overline{P_t})$-free graphs; see Table~\ref{tbl:pt-co-pt-col}.

\begin{table}[h]
\centering
\begin{tabular}{r|cccc}
$t\backslash k$      & $\leq 2$ & $3$     & $\geq 4$\\
\hline
$\leq 5$             & P        & P       & P       \\
$6$                  & P        & P       & ?       \\
$7$                  & P        & P       & ?       \\
$\geq 8$             & P        & ?       & \NP-c   \\\\
\end{tabular}
\caption{\label{tbl:pt-co-pt-col}The complexity of $k$-{\sc Colouring} $(P_t,\overline{P_t})$-free graphs for fixed values of~$k$ and~$t$.
Here, P means polynomial-time solvable and \NP-c means \NP-complete.
The entries in this table originate from Theorem~\ref{thm:4-col-overP8} and the following two results: $k$-{\sc Colouring} is polynomial-time solvable for $P_5$-free graphs for any $k\geq 1$~\cite{HKLSS10} and $3$-{\sc Colouring} is polynomial-time solvable for $P_7$-free graphs~\cite{BCMSZ}.}
\end{table}

\section{Preliminaries}\label{sec:prelim}

Throughout our paper we consider only finite, undirected graphs without multiple edges or self-loops.
The {\em disjoint union} $(V(G)\cup V(H), E(G)\cup E(H))$ of two vertex-disjoint graphs~$G$ and~$H$ is denoted by~$G+\nobreak H$ and the disjoint union of~$r$ copies of a graph~$G$ is denoted by~$rG$.
For a subset $S\subseteq V(G)$, we let~$G[S]$ denote the subgraph of~$G$ {\it induced} by~$S$, which has vertex set~$S$ and edge set $\{uv\; |\; u,v\in S, uv\in E(G)\}$.
If $S=\{s_1,\ldots,s_r\}$ then we may write $G[s_1,\ldots,s_r]$ instead of $G[\{s_1,\ldots,s_r\}]$.
We may also write $G \setminus S$ instead of $G[V(G)\setminus S]$.

For a set of graphs $\{H_1,\ldots,H_p\}$, a graph~$G$ is {\em $(H_1,\ldots,H_p)$-free} if it has no induced subgraph isomorphic to a graph in $\{H_1,\ldots,H_p\}$;
if~$p=1$, we may write $H_1$-free instead of $(H_1)$-free.
For a graph~$G$, the set $N(u)=\{v\in V\; |\; uv\in E\}$ denotes the {\em neighbourhood} of $u\in V(G)$.

The graphs $C_r$, $K_r$, $K_{1,r-1}$ and~$P_r$ denote the cycle, complete graph, star and path on~$r$ vertices, respectively.
The graph $\overline{P_1+2P_2}$ is also known as the $5$-vertex {\em wheel}.
The graphs~$K_3$ and~$K_{1,3}$ are also called the {\em triangle} and {\em claw}, respectively.
A graph~$G$ is a {\em linear forest} if every component of~$G$ is a path (on at least one vertex).
The graph~$S_{h,i,j}$, for $1\leq h\leq i\leq j$, denotes the {\em subdivided claw}, that is, the tree that has only one vertex~$x$ of degree~$3$ and
exactly three leaves, which are of distance~$h$,~$i$ and~$j$ from~$x$, respectively.
Observe that $S_{1,1,1}=K_{1,3}$.

The \emph{chromatic number}~$\chi(G)$ of a graph~$G$ is the smallest integer~$k$ such that~$G$ is $k$-colourable.
The \emph{clique number}~$\omega(G)$ is the size of a largest clique in~$G$.
A graph is {\em bipartite} if its vertex set can be partitioned into two (possibly empty) independent sets.

\section{The Proof of Theorem~\ref{t-self}}\label{s-self}

\noindent
As mentioned in Section~\ref{s-intro}, if we study the complexity of {\sc Colouring} for $(H_1,H_2)$-free classes of graphs that are closed under complementation, it is sufficient to consider the case where~$H_1$ and~$H_2$ are not self-complementary and $H_2=\overline{H_1}$. 
We will give a short proof for this claim.
To do so, we will need the following lemma.

\begin{lemma}[\cite{EHK98,KKTW01,KL07}]\label{lem:col-cycle-free}
For any integer $k\geq 3$, {\sc Colouring} is \NP-complete for $(C_3,C_4,\ldots,C_k)$-free graphs.
\end{lemma}

\medskip
\noindent
\faketheorem{Theorem~\ref{t-self} (restated).}
{\it Let $H_1,\ldots,H_k$ be self-complementary graphs.
Then {\sc Colouring} is polynomial-time solvable for $(H_1,\ldots,H_k)$-free graphs if $H_i \ssi P_4$ for some~$i$ and it is \NP-complete otherwise.}

\begin{proof}
Let~$H$ be a self-complementary graph on~$n$ vertices.
Then~$H$ must have $\frac{1}{2}\binom{n}{2}$ edges.
If $n=1$ then $H=P_1$.
Now~$n$ cannot be~$2$ or~$3$, since $\frac{1}{2}\binom{2}{2}$ and $\frac{1}{2}\binom{3}{2}$ are not integers.
If $n=4$ then $H=P_4$, by inspection.
Suppose $n \geq 5$.
Then $\frac{1}{2}\binom{n}{2} = \frac{n(n-1)}{4} \geq n$, so~$H$ must contain a cycle.
Thus, if~$H$ is a self-complementary graph then it is either an induced subgraph of~$P_4$ or it contains a cycle.

Let $H_1,\ldots,H_k$ be self-complementary graphs.
If $H_i \ssi P_4$ for some~$i$, then the {\sc Colouring} problem for $(H_1,\ldots,H_k)$-free graphs is polynomial-time solvable by Theorem~\ref{t-col-H-free}.
If $H_i \not \ssi P_4$ for all~$i$, then each~$H_i$ must contain a cycle. In that case {\sc Colouring} for $(H_1,\ldots,H_k)$-free graphs is \NP-complete
by Lemma~\ref{lem:col-cycle-free}.\qedllncs
\end{proof}

\section{The Proof of Theorem~\ref{t-mainmain}}\label{a-b}

In this section we prove Theorem~\ref{t-mainmain}.
As part of the proof we first show that 
{\sc Colouring} is \NP-complete for $(P_1+\nobreak 2P_2,\allowbreak \overline{P_1+2P_2},\allowbreak 2P_3,\allowbreak \overline{2P_3},\allowbreak P_6,\overline{P_6})$-free graphs in Section~\ref{sec:hard} and polynomial-time solvable for $(P_2+\nobreak P_3,\allowbreak \overline{P_2+P_3})$-free graphs in Section~\ref{sec:easy}. 

\subsection{The Hardness Result}\label{sec:hard}

In this section we prove the following result.

\begin{theorem}\label{thm:P1+2P2-2P3-P6-and-complements-free-NPC}
{\sc Colouring} is \NP-complete for $(P_1+\nobreak 2P_2,\allowbreak \overline{P_1+2P_2},\allowbreak 2P_3,\allowbreak \overline{2P_3},\allowbreak P_6,\overline{P_6})$-free graphs.
\end{theorem}

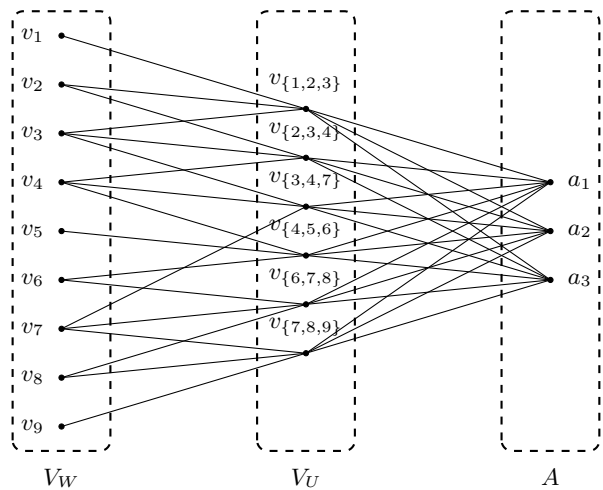
\begin{figure}
\begin{center}
\begin{tikzpicture}[scale=0.65]
\draw[rounded corners, color=black, thick, dashed] (-1,0+0.5) rectangle (1,10-0.5);
\draw[rounded corners, color=black, thick, dashed] (-1+5,0+0.5) rectangle (1+5,10-0.5);
\draw[rounded corners, color=black, thick, dashed] (-1+10,0+0.5) rectangle (1+10,10-0.5);

\draw (0,+0.5) node [label={[label distance=-0pt]-90:$V_W$}] {};
\draw (5,+0.5) node [label={[label distance=-0pt]-90:$V_U$}] {};
\draw (10,+0.5) node [label={[label distance=-0pt]-90:$A$}] {};

\coordinate (v1) at (0,1) ;
\coordinate (v2) at (0,2) ;
\coordinate (v3) at (0,3) ;
\coordinate (v4) at (0,4) ;
\coordinate (v5) at (0,5) ;
\coordinate (v6) at (0,6) ;
\coordinate (v7) at (0,7) ;
\coordinate (v8) at (0,8) ;
\coordinate (v9) at (0,9) ;
\draw [fill=black] (v1) circle (1.5pt) ;
\draw [fill=black] (v2) circle (1.5pt) ;
\draw [fill=black] (v3) circle (1.5pt) ;
\draw [fill=black] (v4) circle (1.5pt) ;
\draw [fill=black] (v5) circle (1.5pt) ;
\draw [fill=black] (v6) circle (1.5pt) ;
\draw [fill=black] (v7) circle (1.5pt) ;
\draw [fill=black] (v8) circle (1.5pt) ;
\draw [fill=black] (v9) circle (1.5pt) ;

\draw (v1) node [label={[label distance=-0pt]+180:$v_9$}] {};
\draw (v2) node [label={[label distance=-0pt]+180:$v_8$}] {};
\draw (v3) node [label={[label distance=-0pt]+180:$v_7$}] {};
\draw (v4) node [label={[label distance=-0pt]+180:$v_6$}] {};
\draw (v5) node [label={[label distance=-0pt]+180:$v_5$}] {};
\draw (v6) node [label={[label distance=-0pt]+180:$v_4$}] {};
\draw (v7) node [label={[label distance=-0pt]+180:$v_3$}] {};
\draw (v8) node [label={[label distance=-0pt]+180:$v_2$}] {};
\draw (v9) node [label={[label distance=-0pt]+180:$v_1$}] {};

\coordinate (w1) at (5,1+1.5) ;
\coordinate (w2) at (5,2+1.5) ;
\coordinate (w3) at (5,3+1.5) ;
\coordinate (w4) at (5,4+1.5) ;
\coordinate (w5) at (5,5+1.5) ;
\coordinate (w6) at (5,6+1.5) ;
\draw [fill=black] (w1) circle (1.5pt) ;
\draw [fill=black] (w2) circle (1.5pt) ;
\draw [fill=black] (w3) circle (1.5pt) ;
\draw [fill=black] (w4) circle (1.5pt) ;
\draw [fill=black] (w5) circle (1.5pt) ;
\draw [fill=black] (w6) circle (1.5pt) ;

\draw (w1) node [label={[label distance=-0pt]+90:$v_{\{7,8,9\}}$}] {};
\draw (w2) node [label={[label distance=-0pt]+90:$v_{\{6,7,8\}}$}] {};
\draw (w3) node [label={[label distance=-0pt]+90:$v_{\{4,5,6\}}$}] {};
\draw (w4) node [label={[label distance=-0pt]+90:$v_{\{3,4,7\}}$}] {};
\draw (w5) node [label={[label distance=-0pt]+90:$v_{\{2,3,4\}}$}] {};
\draw (w6) node [label={[label distance=-0pt]+90:$v_{\{1,2,3\}}$}] {};

\coordinate (x1) at (10,1+3) ;
\coordinate (x2) at (10,2+3) ;
\coordinate (x3) at (10,3+3) ;
\draw [fill=black] (x1) circle (1.5pt) ;
\draw [fill=black] (x2) circle (1.5pt) ;
\draw [fill=black] (x3) circle (1.5pt) ;

\draw (x1) node [label={[label distance=-0pt]+0:$a_3$}] {};
\draw (x2) node [label={[label distance=-0pt]+0:$a_2$}] {};
\draw (x3) node [label={[label distance=-0pt]+0:$a_1$}] {};

\draw (w1) -- (x1);
\draw (w2) -- (x1);
\draw (w3) -- (x1);
\draw (w4) -- (x1);
\draw (w5) -- (x1);
\draw (w6) -- (x1);

\draw (w1) -- (x2);
\draw (w2) -- (x2);
\draw (w3) -- (x2);
\draw (w4) -- (x2);
\draw (w5) -- (x2);
\draw (w6) -- (x2);

\draw (w1) -- (x3);
\draw (w2) -- (x3);
\draw (w3) -- (x3);
\draw (w4) -- (x3);
\draw (w5) -- (x3);
\draw (w6) -- (x3);

\draw (v1) -- (w1);
\draw (v2) -- (w1);
\draw (v3) -- (w1);

\draw (v2) -- (w2);
\draw (v3) -- (w2);
\draw (v4) -- (w2);

\draw (v4) -- (w3);
\draw (v5) -- (w3);
\draw (v6) -- (w3);

\draw (v3) -- (w4);
\draw (v6) -- (w4);
\draw (v7) -- (w4);

\draw (v6) -- (w5);
\draw (v7) -- (w5);
\draw (v8) -- (w5);

\draw (v7) -- (w6);
\draw (v8) -- (w6);
\draw (v9) -- (w6);

\end{tikzpicture}
\caption{\label{fig:G_WU}The graph~$G_{W,U}$ for the instance~$(W,U)$ of {\sc Exact $3$-Cover} where $q=3$, $k=6$, $W=\{1,\ldots,9\}$~and $U=\{
\{1,2,3\},\allowbreak
\{2,3,4\},\allowbreak
\{3,4,7\},\allowbreak
\{4,5,6\},\allowbreak
\{6,7,8\},\allowbreak
\{7,8,9\}
\}$. 
The
edges of the clique~$V_W$ are not~shown.}
\end{center}
\end{figure}

\begin{proof}
\setcounter{ctrclaim}{0}
Let~$k$ and~$q$ be positive integers with $k \geq q$.
Let~$W$ be a set of size~$3q$.
Let~$U$ be a collection of~$k$ subsets of~$W$ each of size~$3$.
An {\it exact $3$-cover} for $(W,U)$ is a set $U' \subseteq U$ of size~$q$ such that every member of~$W$ belongs to one of the subsets in~$U'$.
The \NP-complete problem {\sc Exact $3$-Cover}~\cite{GJ79} is that of determining if such a set~$U'$ exists. 
To prove the theorem, we describe a reduction from {\sc Exact $3$-Cover} to {\sc Colouring} for $(P_1+\nobreak 2P_2,\allowbreak \overline{P_1+2P_2},\allowbreak 2P_3,\allowbreak \overline{2P_3},\allowbreak P_6,\overline{P_6})$-free graphs.

Given an instance $(W,U)$ of {\sc Exact $3$-Cover}, we construct the graph~$G_{W,U}$ as follows (see~\figurename~\ref{fig:G_WU} for an example):
\begin{itemize}
\item Introduce a set of vertices $V_W = \{v_w \mid w \in W\}$ which forms a clique in~$G_{W,U}$. 
\item Introduce a set of vertices $V_U = \{v_u \mid u \in U\}$ which forms an independent set in~$G_{W,U}$.
\item Add an edge from $v_w \in V_W$ to $v_u \in V_U$ if and only if $w \in u$.
\item Introduce a set of vertices $A = \{a_i \mid 1 \leq i \leq k-q\}$ which forms an independent set in~$G_{W,U}$.
\item Add an edge from each vertex of~$A$ to each vertex of~$V_U$.
\end{itemize}

\clm{\label{clm3:reduction}The vertices of $G_{W,U}$ can be covered by at most~$k$ pairwise vertex-disjoint cliques if and only if~$U$ contains an exact $3$-cover~$U'$.}
First suppose that~$U$ contains an exact $3$-cover $U'$.
Let $V_{U'} = \{v_u \in V_U \mid u \in U'\}$.
For each $u \in U'$, let~$K^u$ be the clique on four vertices containing~$v_u$ and its three neighbours in~$V_W$; there are~$q$ such cliques.
Form a perfect matching between the vertices of $V_U \setminus V_{U'}$ and the vertices of~$A$, that is, a collection of $k-q$ cliques on two vertices.
Together, these~$k$ pairwise vertex-disjoint cliques cover~$V(G_{W,U})$.

Now suppose that the vertices of~$G_{W,U}$ can be covered by at most~$k$ pairwise vertex-disjoint cliques.
As~$V_U$ is independent and $|V_U|=k$, we have exactly~$k$ cliques, and each of them contains exactly one vertex of~$V_U$. 
Hence, as~$A$ is an independent set and each vertex in~$A$ is only adjacent to the vertices of~$V_U$, each of the $k-q$ vertices of~$A$ must be contained in a clique of size~$2$ that consists of a vertex of~$A$ and a vertex of~$V_U$.
There are~$q$ other cliques, which, as we deduced, also contain exactly one vertex of~$V_U$.
Each vertex of~$V_W$ must be in one of these cliques.
As each vertex in~$V_U$ has exactly three neighbours in~$V_W$ and there are~$3q$ vertices in~$V_W$, this means that each of these cliques must contain four vertices, consisting of one vertex of~$V_U$ and three vertices of~$V_W$.
Hence, if we let $u \in U$ belong to~$U'$ whenever~$v_u$ is in a clique of size~$4$, then~$U'$ forms an exact $3$-cover.
This completes the proof of Claim~\ref{clm3:reduction}.

\medskip
\noindent
We have shown a reduction from {\sc Exact $3$-Cover} to the problem of finding~$k$ pairwise vertex-disjoint cliques that cover $V(G_{W,U})$.
The latter problem is equivalent to finding a $k$-colouring of~$\overline{G_{W,U}}$ (note that~$k$ is part of the input of {\sc Exact $3$-Cover}).
It remains to show that~$\overline{G_{W,U}}$ is $(P_1+\nobreak 2P_2,\allowbreak \overline{P_1+2P_2},\allowbreak 2P_3,\allowbreak \overline{2P_3},\allowbreak P_6,\overline{P_6})$-free.

Since a graph~$G$ is $(P_1+\nobreak 2P_2,\allowbreak \overline{P_1+2P_2},\allowbreak 2P_3,\allowbreak \overline{2P_3},\allowbreak P_6,\overline{P_6})$-free if and only if~$\overline{G}$ is, it is sufficient to show that~$G_{W,U}$ is $(P_1+\nobreak 2P_2,\allowbreak \overline{P_1+2P_2},\allowbreak 2P_3,\allowbreak \overline{2P_3},\allowbreak P_6,\overline{P_6})$-free.
We do this in Claims~\ref{clm3:P1+2P2-free}-\nobreak\ref{clm3:co-P6-free}, but first we prove a useful observation. 

\clm{\label{clm3:jfour}Let~$J$ be an induced subgraph of~$G_{W,U}$ that is the complement of a linear forest on six vertices.
Then~$J$ contains at least four vertices of~$V_W$.}
As a linear forest contains no triangle, $J$ contains no independent set on three vertices.
So~$J$ cannot contain more than two vertices from either of the independent sets~$V_U$ or~$A$.
Thus it contains vertices of~$V_W$.
But then it cannot contain two vertices from~$A$ as, combined with a vertex of~$V_W$ this would again induce an independent set of size~$3$.
Hence there are at least three vertices of~$V_W$ in~$J$.
But this implies that not even one vertex of~$A$ belongs to~$J$ (as it would have three non-neighbours and every vertex of~$J$ is adjacent to all but at most two of the others).
This completes the proof of Claim~\ref{clm3:jfour}.

\clm{\label{clm3:P1+2P2-free}$G_{W,U}$ is $(P_1+\nobreak 2P_2)$-free.}
Suppose for contradiction that there is an induced $P_1+\nobreak 2P_2$ in~$G_{W,U}$.
Since~$V_U$ and~$A$ are independent sets, every~$P_2$ in~$G_{W,U}$ must either have two vertices in~$V_W$ or one vertex in~$V_U$ and the other in~$V_W$ or~$A$.
Since that~$V_W$ is a clique and every vertex in~$A$ is adjacent to every vertex in~$V_U$, one of the~$P_2$'s must have both its vertices in~$V_W$ and the other must have one vertex in~$V_U$ and the other in~$A$.
But every vertex in $G_{W,U}$ has a neighbour in such a~$2P_2$, so no induced $P_1+\nobreak 2P_2$ can exist.
This completes the proof of Claim~\ref{clm3:P1+2P2-free}.

\clm{\label{clm3:co-P1+2P2-free}$G_{W,U}$ is $\overline{P_1+2P_2}$-free.}
Suppose for contradiction that there is an induced $\overline{P_1+2P_2}$ in~$G_{W,U}$.
Note this subgraph consists of a~$C_4$, which we denote~$C$, plus an additional vertex, which we denote~$z$, that is adjacent to every vertex in~$C$.
If $z \in A$, then the vertices of~$C$ are all in~$V_U$, and if $z \in V_U$, the vertices of the~$C$ are either all in~$A$ or all in~$V_W$.
Neither is possible, so we must have $z \in V_W$.
Therefore none of the vertices of~$C$ is in~$A$.
At most two of the vertices of~$C$ are in each of~$V_U$ and~$V_W$ as~$C$ contains neither an independent set nor a clique on three vertices.
So~$C$ contains exactly two vertices from each of~$V_U$ and~$V_W$, but the pair from~$V_U$ must be non-adjacent in~$C$ and the pair from~$V_W$ must be adjacent in~$C$.
This contradiction completes the proof of Claim~\ref{clm3:co-P1+2P2-free}.

\clm{\label{clm3:2P3-free}$G_{W,U}$ is $2P_3$-free.} 
Suppose for contradiction that there is an induced~$2P_3$ in~$G_{W,U}$.
Denote this graph by~$P$.
As~$V_W$ is a clique,~$P$ contains at most two of its vertices.
If it contains exactly two, then~$P$ must also contain a vertex $v_u \in V_U$ such that the three vertices together induce a~$P_3$.
But then~$v_u$ cannot be adjacent to any other vertex of~$P$ so the remaining three vertices must all belong to~$V_U$, a contradiction as~$V_U$ is independent.
So~$P$ must contain at least five vertices of~$A$ and~$V_U$.
They cannot all belong to one of these two sets as~$P$ has no independent set of size~$5$, but if~$P$ contains vertices of both sets then there must be a vertex of degree~$3$.
This contradiction proves that~$P$ does not exist.
This completes the proof of Claim~\ref{clm3:2P3-free}.

\clm{\label{clm3:co-2P_3-free}$G_{W,U}$ is $\overline{2P_3}$-free.} 
Suppose for contradiction that there is an induced~$\overline{2P_3}$ in~$G$.
By Claim~\ref{clm3:jfour}, it contains four vertices of~$V_W$.
We note that~$\overline{2P_3}$ contains exactly one induced~$K_4$ and that the other two vertices are adjacent to exactly two vertices in the clique (so must both be in~$V_U$) and to each other (so cannot both be in~$V_U$).
This contradiction completes the proof of Claim~\ref{clm3:co-2P_3-free}.

\clm{\label{clm3:P6-free}$G_{W,U}$ is $P_6$-free.}
Suppose for contradiction that there is an induced~$P_6$ in~$G_{W,U}$.
Denote this path by~$P$.
As~$V_W$ is a clique,~$P$ contains at most two of its vertices.
We see that~$P$ cannot contain exactly four vertices of either~$V_U$ or~$A$ as they would form an independent set of size~$4$. $P$ cannot contain exactly three vertices of either~$V_U$ or~$A$ as then there would be a vertex in the other of the two sets of degree~$3$. Finally, $P$ cannot contain at least two vertices from each of~$V_U$ and~$A$ as then it would contain an induced~$C_4$.
These contradictions prove that~$P$ does not exist.
This completes the proof of Claim~\ref{clm3:P6-free}.

\clm{\label{clm3:co-P6-free}$G_{W,U}$ is $\overline{P_6}$-free.}
Suppose for contradiction that there is an induced~$\overline{P_6}$ in~$G$.
By Claim~\ref{clm3:jfour}, it contains four vertices of~$V_W$.
This contradiction -- a clique on four vertices is not an induced subgraph of~$\overline{P_6}$ -- completes the proof of Claim~\ref{clm3:co-P6-free}.\qedllncs
\end{proof}

\subsection{A Tractable Case}\label{sec:easy}

We show that {\sc Colouring} is polynomial-time solvable for $(P_2+\nobreak P_3,\allowbreak \overline{P_2+P_3})$-free graphs.
We first introduce some additional terminology.
The {\em clique-width}~$\cw(G)$ of a graph~$G$ is the minimum number of labels needed to construct~$G$ by using the following four operations:
\begin{enumerate}
\item creating a new graph consisting of a single vertex~$v$ with label~$i$;
\item taking the disjoint union of two labelled graphs~$G_1$ and~$G_2$;
\item joining each vertex with label~$i$ to each vertex with label~$j$ ($i\neq j$);
\item renaming label~$i$ to~$j$.
\end{enumerate}

\noindent
A class of graphs~${\cal G}$ has {\em bounded} clique-width if there is a constant~$c$ such that the clique-width of every graph in~${\cal G}$ is at most~$c$; otherwise the clique-width  is {\em unbounded}.
For an induced subgraph~$G'$ of a graph~$G$, the {\em subgraph complementation} operation (acting on~$G$ with respect to~$G'$) replaces every edge present in~$G'$ by a non-edge, and vice versa.
Similarly, for two disjoint vertex subsets~$S$ and~$T$ in~$G$, the {\em bipartite complementation} operation with respect to~$S$ and~$T$ acts on~$G$ by replacing every edge with one end-vertex in~$S$ and the other one in~$T$ by a non-edge and vice versa.
Let $k\geq 0$ be a constant and let~$\gamma$ be some graph operation.
We say that a graph class~${\cal G'}$ is {\em $(k,\gamma)$-obtained} from a graph class~${\cal G}$ if the following two conditions hold:
\begin{enumerate}
\item every graph in~${\cal G'}$ is obtained from a graph in~${\cal G}$ by performing~$\gamma$ at most~$k$ times, and
\item for every $G\in {\cal G}$ there exists at least one graph in~${\cal G'}$ that is obtained from~$G$ by performing~$\gamma$ at most~$k$ times.
\end{enumerate}

\noindent
We say that~$\gamma$ {\em preserves} boundedness of clique-width if for any finite constant~$k$ and any graph class~${\cal G}$, any graph class~${\cal G}'$ that is $(k,\gamma)$-obtained from~${\cal G}$ has bounded clique-width if and only if~${\cal G}$ has bounded clique-width.
\begin{enumerate}[\bf F{a}ct 1.]
\item \label{fact:del-vert} Vertex deletion preserves boundedness of clique-width~\cite{LR04}.\\[-1em]

\item \label{fact:comp} Subgraph complementation preserves boundedness of clique-width~\cite{KLM09}.\\[-1em]

\item \label{fact:bip} Bipartite complementation preserves boundedness of clique-width~\cite{KLM09}.\\[-1em]

\end{enumerate}

\noindent
Two (non-adjacent) vertices are {\em false twins} if they have the same neighbourhood. 
We will use the following two lemmas, which are readily seen.

\begin{lemma}\label{lem:atmost-2}
The clique-width of a graph of maximum degree at most~$2$ is at most~$4$.
\end{lemma}

\begin{lemma}\label{lem:false-twin}
If a vertex~$x$ in a graph~$G$ has a false twin then $\cw(G)=\cw(G \setminus \{x\})$.
\end{lemma}

The following lemma for $(P_2+\nobreak P_3)$-free bipartite graphs follows from a result of Lozin~\cite{Lozin2002}, who proved that the superclass of $S_{1,2,3}$-free bipartite graphs has bounded clique-width (see~\cite{DP14} for a full classification of the boundedness of clique-width of $H$-free bipartite graphs).

\begin{lemma}[\cite{Lozin2002}]\label{lem:bipartite}
The class of $(P_2+\nobreak P_3)$-free bipartite graphs has bounded clique-width.
\end{lemma}

A graph~$G$ is {\em perfect} if $\chi(H)=\omega(H)$ for every induced subgraph~$H$ of~$G$.
Besides the above lemmas we will also apply the following three well-known theorems.

\begin{theorem}[\cite{CRST06}]\label{thm:spgt}
A graph is perfect if and only if it is $C_r$-free and $\overline{C_r}$-free
for every odd $r \geq 5$.
\end{theorem}

\begin{theorem}[\cite{GLS84}]\label{t-perfect-colouring}
{\sc Colouring} is polynomial-time solvable on perfect graphs.
\end{theorem}

\begin{theorem}[\cite{KR03b,OS06}]\label{thm:colouring-for-bdd-cw}
For any constant~$c$, {\sc Colouring} is polynomial-time solvable on graphs of clique-width at most~$c$.
\end{theorem}

If~$G$ is a graph, then a (possibly empty) set $C \subseteq V(G)$ is a {\em clique separator} if~$C$ is a clique and $G \setminus C$ is disconnected.
A graph~$G$ is an {\em atom} if it has no clique separator.
Due to the following result of Tarjan, when considering the {\sc Colouring} problem in a hereditary class, we only need to consider atoms in the class.

\begin{lemma}[\cite{Tarjan85}]\label{lem:colouring-no-clique-cutset}
If {\sc Colouring} is polynomial-time solvable on atoms in an hereditary class~${\cal G}$, then it is polynomial-time solvable on all graphs in~${\cal G}$.
\end{lemma}

\noindent
Let~$X$ be a set of vertices in a graph $G=(V,E)$.
A vertex $y\in V\setminus X$ is {\em complete} (resp. {\em anti-complete}) to~$X$ if it is adjacent (resp. non-adjacent) to every vertex of~$X$.
A set of vertices $Y\subseteq V\setminus X$ is {\em complete} (resp. {\em anti-complete}) to~$X$ if every vertex in~$Y$ is complete (resp. anti-complete) to~$X$.
If~$X$ and~$Y$ are disjoint sets of vertices in a graph, we say that the edges between these two sets form a {\em matching} if each vertex in~$X$ has at most one neighbour in~$Y$ and vice versa (if each vertex has exactly one such neighbour, we say that the matching is {\em perfect}).
Similarly, the edges between these sets form a {\em co-matching} if each vertex in~$X$ has at most one non-neighbour in~$Y$ and vice versa.

In Theorem~\ref{thm:P2P_3-co-P_2+P_3-free} we present an algorithm for {\sc Colouring} restricted to $(P_2+\nobreak P_3,\allowbreak \overline{P_2+P_3})$-free graphs.
Our algorithm first tries to reduce to perfect graphs, in which case we can use Theorem~\ref{t-perfect-colouring}.
If a $(P_2+\nobreak P_3,\allowbreak \overline{P_2+P_3})$-free graph is not perfect, then we show that it must contain an induced~$C_5$.
The clique-width of such graphs is not bounded 
(we can construct a $(P_2+\nobreak P_3,\overline{P_2+P_3})$-free graph of arbitrarily large clique-width by taking a split graph of arbitrarily large clique-width~\cite{MR99} with clique~$C$ and independent set~$I$ and adding five new vertices that form an induced~$C_5$
and that are complete to~$C$ but anti-complete to~$I$).
Therefore, we do some pre-processing to simplify the graph.
This enables us to bound the clique-width so that we may then apply Theorem~\ref{thm:colouring-for-bdd-cw}.
Our general scheme for bounding the clique-width is to partition the remaining vertices into sets according to their neighbourhood in the~$C_5$ and then investigate the possible edges both inside these sets and between them.
This enables us to use graph operations (see also Facts~\ref{fact:del-vert}--\ref{fact:bip}) that do not change the clique-width by ``too much'' to partition the input graph into disjoint pieces known to have bounded clique-width.

\begin{theorem}\label{thm:P2P_3-co-P_2+P_3-free}
{\sc Colouring} is polynomial-time solvable for $(P_2+\nobreak P_3,\overline{P_2+P_3})$-free graphs.
\end{theorem}

\begin{proof}
\setcounter{ctrclaim}{0}
Let~$G$ be a $(P_2+\nobreak P_3,\overline{P_2+P_3})$-free graph.
By Lemma~\ref{lem:colouring-no-clique-cutset} we may assume that~$G$ is an atom and so has no clique separator (we will use this assumption in the proof of Claim~\ref{clm:V_0-V_12345-triv}).
We first test in $O(|V(G)|^5)$ time whether~$G$ contains an induced~$C_5$.
If not, then~$G$ is $C_5$-free.
Since $\overline{C_5}=C_5$, $G$ is also $\overline{C_5}$-free.
Since~$G$ is $(P_2+\nobreak P_3)$-free, $G$ is $C_k$-free for all $k \geq 7$.
Since~$G$ is $(\overline{P_2+\nobreak P_3})$-free, $G$ is $\overline{C_k}$-free for all $k \geq 7$.
By Theorem~\ref{thm:spgt}, this means that~$G$ is a perfect graph.
Therefore, by Theorem~\ref{t-perfect-colouring}, we can solve {\sc Colouring} in polynomial time in this case.

\begin{sloppypar}
Now assume our algorithm found an induced $5$-vertex cycle~$C$ with vertices $v_1,v_2,v_3,v_4,v_5$, in that order.
For $S \subseteq \{1,\ldots,5\}$, let~$V_S$ be the set of vertices $x \in V(G) \setminus V(C)$ such that $N(x)\cap V(C)=\{v_i \;|\; i \in S\}$. 
A set~$V_S$ is {\em large} if it contains at least three vertices, otherwise it is {\em small}.
\end{sloppypar}

To ease notation, in the remainder of the proof, subscripts on vertex sets should be interpreted modulo~$5$
and whenever possible we will write~$V_i$ instead of~$V_{\{i\}}$ and~$V_{i,j}$ instead of~$V_{\{i,j\}}$ and so on.

We start with three structural claims.

\clm{\label{clm:V0-indep}$V_\emptyset$ is an independent set.}
If $x,y \in V_\emptyset$ are adjacent then $G[x,y,v_1,v_2,v_3]$ is a $P_2+\nobreak P_3$, a contradiction.
Therefore $V_\emptyset$ is an independent set. 
This completes the proof of Claim~\ref{clm:V0-indep}.

\begin{sloppypar}
\clm{\label{clm:V1cupV12-small}$|V_i \cup V_{i,i+1}| \leq 1$ and $|V_{i+1} \cup V_{i,i+1}| \leq 1$ for $i \in \{1,2,3,4,5\}$.}
Suppose for contradiction that $x,y \in V_1 \cup V_{1,2}$.
If~$x$ and~$y$ are adjacent then $G[x,y,v_3,v_4,v_5]$ is a $P_2+\nobreak P_3$ and if~$x$ and~$y$ are non-adjacent then $G[v_3,v_4,x,v_1,y]$ is a $P_2+\nobreak P_3$.
This contradiction implies that $|V_1 \cup V_{1,2}| \leq 1$.
Claim~\ref{clm:V1cupV12-small} follows by symmetry.
\end{sloppypar}

\clm{\label{clm:V13-indep}$V_{i,i+2}$ is an independent set for $i \in \{1,2,3,4,5\}$.}
Suppose for contradiction that $x,y \in V_{1,3}$ are adjacent.
Then $G[v_1,v_3,x,v_2,y]$ is a $\overline{P_2+P_3}$, a contradiction.
Therefore $V_{1,3}$ is an independent set.
Claim~\ref{clm:V13-indep} follows by symmetry.

\medskip
\noindent
Since $\overline{C_5}=C_5$, it follows that~$\overline{G}$ also contains a~$C_5$, namely on the vertices $v_1$, $v_3$, $v_5$, $v_2$ and~$v_4$ in that order.
Therefore~$\overline{G}$ is also a $(P_2+\nobreak P_3,\allowbreak \overline{P_2+P_3})$-free graph containing an induced~$C_5$.
As a result, we immediately obtain the following claims as corollaries of Claims~\ref{clm:V0-indep}--\ref{clm:V13-indep}.

\clm{\label{clm:V12345-clique}$V_{1,2,3,4,5}$ is a clique.}
\shortclm{\label{clm:V2345cupV245-small}$|V_{i,i+1,i+2,i+3} \cup V_{i,i+1,i+3}| \leq 1$ and $|V_{i,i+1,i+2,i+3} \cup V_{i,i+2,i+3}| \leq 1$ for $i \in \{1,2,3,4,5\}$.}\\
\shortclm{\label{clm:V234-clique}$V_{i,i+1,i+2}$ is a clique for $i \in \{1,2,3,4,5\}$.}

\medskip
\noindent
Let~$I$ be the set of vertices in~$V_\emptyset$ that have a non-neighbour in~$V_{1,2,3,4,5}$.

\clm{\label{clm:V_0-V_12345-triv}$G$ is $k$-colourable if and only if $G \setminus I$ is $k$-colourable.}
The ``only if'' direction is trivial.
Suppose~$G \setminus I$ is $k$-colourable. Fix a $k$-colouring~$c$ of $G \setminus I$.
We will show how to extend $c$ to all vertices of~$G$.

We may assume~$V_\emptyset$ and~$V_{1,2,3,4,5}$ are non-empty, otherwise $I=\emptyset$ and the claim follows trivially.
Note that
$V_\emptyset$ is an independent set by Claim~\ref{clm:V0-indep} and that~$V_{1,2,3,4,5}$ is a clique by Claim~\ref{clm:V12345-clique}.
By assumption, $V_{1,2,3,4,5}$ is not a clique separator in~$G$.
Since~$V_\emptyset$ is an independent set, every vertex of~$V_\emptyset$ must have a neighbour in $V(G) \setminus (V_\emptyset \cup V_{1,2,3,4,5} \cup V(C))$ that is adjacent to some vertex of the cycle~$C$.
Let~$A$ be the set of vertices outside~$V_{1,2,3,4,5}$ that have a neighbour in~$V_\emptyset$.
Note that~$A$ contains no vertex of~$V_\emptyset$, as~$V_\emptyset$ is an independent set.
Hence, every vertex of~$A$ is adjacent to least one vertex of $C$
(but not to all vertices of~$C$).

We claim that~$A$ is complete to~$V_{1,2,3,4,5}$.
Indeed, suppose for contradiction that~$x \in A$ is non-adjacent to~$y \in V_{1,2,3,4,5}$.
Let~$s$ be a vertex in~$V_\emptyset$ that is a neighbour of~$x$. 
Without loss of generality, assume that~$x$ is adjacent to~$v_1$.
Then~$x$ is adjacent to~$v_3$ or~$v_4$, otherwise $G[v_3,v_4,s,x,v_1]$ would be a $P_2+\nobreak P_3$, a contradiction.
Without loss of generality, assume that~$x$ is adjacent to~$v_3$.
Then~$x$ is adjacent to~$v_2$, otherwise $G[v_1,v_3,y,x,v_2]$ would be a $\overline{P_2+P_3}$, a contradiction.
Thus~$x$ is adjacent to~$v_4$ or~$v_5$, otherwise $G[v_4,v_5,s,x,v_2]$ would be a $P_2+\nobreak P_3$, a contradiction.
Without loss of generality, assume that~$x$ is adjacent to~$v_4$.
Now~$x$ must be adjacent to~$v_5$, otherwise $G[v_1,v_4,y,x,v_5]$ would be a $\overline{P_2+P_3}$, a contradiction.
Therefore $x \in V_{1,2,3,4,5}$, a contradiction.
It follows that~$A$ must indeed be complete to~$V_{1,2,3,4,5}$.

Let $x \in I$.
By definition of~$I$, there is a vertex $y \in V_{1,2,3,4,5}$ that is non-adjacent to~$x$.
We claim that~$c(y)$ can also be used to colour~$x$.
Indeed, every neighbour of~$x$ lies in $A \cup V_{1,2,3,4,5} \setminus \{y\}$.
Since~$V_{1,2,3,4,5}$ is a clique and~$A$ is complete to~$V_{1,2,3,4,5}$, no vertex of $A \cup V_{1,2,3,4,5} \setminus \{y\}$ is coloured with the colour~$c(y)$.
We may therefore colour~$x$ with the colour~$c(y)$.
Repeating this process, we can extend the $k$-colouring of $G \setminus I$ to a $k$-colouring of~$G$.
This completes the proof of Claim~\ref{clm:V_0-V_12345-triv}.

\medskip
\noindent
By Claim~\ref{clm:V_0-V_12345-triv}, our algorithm may safely remove all vertices of~$I$ from~$G$ without changing the chromatic number.
Hence we may assume that the following claim holds.

\clmnonewline{\label{clm:V_0-V_12345-trivb}$V_\emptyset$ is complete to~$V_{1,2,3,4,5}$.}

\medskip
\noindent
Our algorithm has finished the pre-processing. We are now ready to prove, through a series of further claims, that~$G$ has {\bf bounded clique-width}.
Hence our algorithm can colour~$G$ in polynomial time by Theorem~\ref{thm:colouring-for-bdd-cw}. 

\clm{\label{clm:V_13-or-V123-small}Suppose $S,T \subseteq \{1,2,3,4,5\}$ with $|S|=2$ and $|T|=3$. Then at least one of~$V_S$ and~$V_T$ is small.}
Suppose for contradiction that there are large sets~$V_S$ and~$V_T$ for $S,T \subseteq \{1,2,3,4,5\}$ with $|S|=2$ and $|T|=3$.
By Claim~\ref{clm:V1cupV12-small} and symmetry, we may assume that $S=\{1,3\}$.
By Claim~\ref{clm:V2345cupV245-small} and symmetry, we may assume that $T = \{i,i+1,i+2\}$ for some $i \in \{1,2,3\}$.
We consider each of these cases in turn.

Suppose $x \in V_{1,3}$ is non-adjacent to $y \in V_{1,2,3}$.
Then $G[v_1,v_3,v_2,x,y]$ is a $\overline{P_2+P_3}$, a contradiction.
Therefore~$V_{1,3}$ is complete to~$V_{1,2,3}$.
Suppose $x,x' \in V_{1,3}$ and $y \in V_{1,2,3}$.
Then~$y$ must be adjacent to both~$x$ and~$x'$.
By Claim~\ref{clm:V13-indep}, $x$ is non-adjacent to~$x'$.
Therefore $G[v_4,v_5,x,y,x']$ is a $P_2+\nobreak P_3$, a contradiction.
It follows that if~$V_{1,3}$ is large then~$V_{1,2,3}$ is empty.

Suppose $x \in V_{1,3}$ is adjacent to $y \in V_{2,3,4}$.
Then $G[x,v_2,v_3,v_1,y]$ is a $\overline{P_2+P_3}$, a contradiction.
Therefore~$V_{1,3}$ is anti-complete to~$V_{2,3,4}$.
Suppose $x,x' \in V_{1,3}$ and $y \in V_{2,3,4}$.
Then~$y$ is non-adjacent to both~$x$ and~$x'$.
By Claim~\ref{clm:V13-indep}, $x$ is non-adjacent to~$x'$.
Therefore $G[v_4,y,x,v_1,x']$ is a $P_2+\nobreak P_3$, a contradiction.
It follows that if~$V_{1,3}$ is large then~$V_{2,3,4}$ is empty.

Suppose $x \in V_{1,3}$ is non-adjacent to $y \in V_{3,4,5}$.
Then $G[v_4,y,v_2,v_1,x]$ is a $P_2+\nobreak P_3$, a contradiction.
Therefore~$V_{1,3}$ is complete to~$V_{3,4,5}$.
Suppose $x,x' \in V_{1,3}$ and $y \in V_{3,4,5}$.
Then~$y$ must be adjacent to both~$x$ and~$x'$.
By Claim~\ref{clm:V13-indep}, $x$ is non-adjacent to~$x'$.
Therefore $G[x,x',y,v_1,v_3]$ is a $\overline{P_2+\nobreak P_3}$, a contradiction.
It follows that if~$V_{1,3}$ is large then~$V_{3,4,5}$ is empty.
This completes the proof of Claim~\ref{clm:V_13-or-V123-small}.

\medskip
\noindent
Claim~\ref{clm:V_13-or-V123-small}
implies that we can have large sets~$V_S$ with $|S|=2$ or large sets~$V_T$ with $|T|=3$ (or neither), but not both.
Again, since $\overline{C_5}=C_5$, the graph~$\overline{G}$ contains a~$C_5$ on vertices $v_1,v_3,v_5,v_2,v_4$, in that order.
A vertex in~$G$ with exactly two (resp. three) neighbours in~$C$ will have exactly three (resp. two) neighbours in this new cycle in~$\overline{G}$.
Therefore, if there is a large set~$V_T$ with $|T|=3$ then~$V_S$ is small whenever $|S|=2$. In this case we replace~$G$ by~$\overline{G}$ and replace~$C$ by this new cycle (we may do this by Fact~\ref{fact:comp}).
We may therefore assume that~$V_T$ is small for all~$T$ with $|T|=3$.
Note that after doing this we no longer have the symmetry between the situation for~$G$ and that for~$\overline{G}$.
Indeed, in Claim~\ref{clm:V_0-V_12345-triv} we showed that because~$G$ has no clique separator, we may assume that~$V_\emptyset$ is complete to~$V_{1,2,3,4,5}$ (Claim~\ref{clm:V_0-V_12345-trivb}).
However, we cannot guarantee that~$\overline{G}$ has no clique separator.
Therefore, if we do complement~$G$ above, then the sets~$V_\emptyset$ and~$V_{1,2,3,4,5}$ will be swapped and will become anti-complete to each other, 
instead of complete to each other.

Suppose~$V_S$ is small for some set $S \subseteq \{1,2,3,4,5\}$.
Then, by Fact~\ref{fact:del-vert}, we may delete the vertices of~$V_S$.
We may therefore assume that for each $S \subseteq \{1,2,3,4,5\}$, $V_S$ is either large or empty.
By Claims~\ref{clm:V1cupV12-small} and~\ref{clm:V2345cupV245-small} 
and our assumption that~$V_T$ is small for all~$T$ with $|T|=3$
it follows that the only sets~$V_S$ that can be large are $V_\emptyset$, $V_{1,2,3,4,5}$ and~$V_{i,i+2}$ for $i \in \{1,2,3,4,5\}$.

\begin{sloppypar}
\clm{\label{clm:V13-V12345-comp}$V_{1,2,3,4,5}$ is complete to~$V_{i,i+2}$ for all $i \in \{1,2,3,4,5\}$.}
Suppose for contradiction that $x \in V_{1,3}$ is non-adjacent to $y \in V_{1,2,3,4,5}$.
Then $G[v_1,v_3,v_2,x,y]$ is a $\overline{P_2+\nobreak P_3}$, a contradiction.
Claim~\ref{clm:V13-V12345-comp} follows by symmetry.
\end{sloppypar}

\begin{sloppypar}
\clm{\label{clm:V12345-empty}We may assume  
that~$V_{1,2,3,4,5}$ is empty.}
By Claim~\ref{clm:V13-V12345-comp}, $V_{1,2,3,4,5}$ is complete to each set~$V_{i,i+2}$.
By Claim~\ref{clm:V_0-V_12345-trivb}, $V_\emptyset$ is complete to~$V_{1,2,3,4,5}$ or, if we complemented~$G$ then it is anti-complete to~$V_{1,2,3,4,5}$.
By Fact~\ref{fact:bip}, 
we may apply a bipartite complementation between~$V_{1,2,3,4,5}$ and $V(C) \cup \bigcup_{i \in \{1,2,3,4,5\}} V_{i,i+2}$.
If~$V_\emptyset$ is complete to~$V_{1,2,3,4,5}$ then we apply a bipartite complementation between~$V_{1,2,3,4,5}$ and~$V_\emptyset$.
These operations disconnect $G[V_{1,2,3,4,5}]$ from the rest of the graph.
By Claim~\ref{clm:V12345-clique}, $G[V_{1,2,3,4,5}]$ is a complete graph, so it has clique-width at most~$2$.
We may therefore assume that~$V_{1,2,3,4,5}$ is empty.
This completes the proof of Claim~\ref{clm:V12345-empty}.
\end{sloppypar}

\clm{\label{clm:Vii+2-V0-matching}For $i \in \{1,2,3,4,5\}$, the edges between~$V_{i,i+2}$ and~$V_\emptyset$ form a matching.}
By symmetry, it suffices to prove the claim for $i=1$.
Note that~$V_\emptyset$ and~$V_{1,3}$ are independent sets by Claims~\ref{clm:V0-indep} and~\ref{clm:V13-indep}, respectively.
If the claim is false then there must be a vertex~$y$ in one of these sets that has two neighbours~$x$ and~$x'$ in the other set.
Then~$x$ and~$x'$ are non-adjacent, so $G[v_4,v_5,x,y,x']$ is a $P_2+\nobreak P_3$.
This contradiction completes the proof of Claim~\ref{clm:Vii+2-V0-matching}.

\clm{\label{clm:v0-special-vii+2-nbrs}For distinct $i,j \in \{1,2,3,4,5\}$, if $x \in V_{i,i+2}$ is adjacent to $y \in V_\emptyset$ and $z \in V_{j,j+2}$ is non-adjacent to~$y$ then~$x$ is adjacent to~$z$.}
Suppose, for contradiction, that the claim is false.
By symmetry, we may assume that $i=1$, $j \in \{2,3\}$.
Indeed, suppose $x \in V_{1,3}$ is adjacent to $y \in V_\emptyset$ and $z \in V_{j,j+2}$ is non-adjacent to~$x$ and~$y$.
Then $G[x,y,v_5,v_4,z]$ or $G[x,y,v_4,v_5,z]$ is a $P_2+\nobreak P_3$ if $j=2$ or $j=3$, respectively.
This contradiction completes the proof of Claim~\ref{clm:v0-special-vii+2-nbrs}.

\clm{\label{clm:v0-empty}We may assume
that~$V_\emptyset$ is empty.}
Recall that~$V_{i,i+2}$ is an independent set for $i \in \{1,2,3,4,5\}$, by Claim~\ref{clm:V13-indep}.
For $i \in \{1,2,3,4,5\}$, let~$V_{i,i+2}'$ be the set of vertices in~$V_{i,i+2}$ that have a neighbour in~$V_\emptyset$ and let $G'= G[V_\emptyset \cup \bigcup_{i \in \{1,2,3,4,5\}} V_{i,i+2}']$.
By Claim~\ref{clm:v0-special-vii+2-nbrs}, for distinct $i,j \in \{1,2,3,4,5\}$, we find that~$V_{i,i+2}'$ is complete to $V_{j,j+2} \setminus V_{j,j+2}'$.
By Fact~\ref{fact:bip} we may apply a bipartite complementation between every such pair of sets and also between $\{v_i,v_{i+2}\}$ and~$V_{i,i+2}'$ for $i \in \{1,2,3,4,5\}$.
This disconnects~$G'$ from the rest of the graph. 
We now prove that~$G'$ has bounded clique-width.

Claim~\ref{clm:Vii+2-V0-matching} implies that for $i \in \{1,2,3,4,5\}$, every vertex in~$V_\emptyset$ has at most one neighbour in~$V_{i,i+2}'$ and every vertex in~$V_{i,i+2}'$ has exactly one neighbour in~$V_\emptyset$.
Furthermore, for distinct $i,j \in \{1,2,3,4,5\}$, if $x \in V_{i,i+2}'$ and $y \in V_{j,j+2}'$ have different neighbours in~$V_\emptyset$ then~$x$ is adjacent to~$y$ by Claim~\ref{clm:v0-special-vii+2-nbrs}.
Now let~$G''$ be the graph obtained from~$G'$ by applying a bipartite complementation between~$V_{i,i+2}'$ and~$V_{j,j+2}'$ for each pair of distinct $i,j \in \{1,2,3,4,5\}$.
Then, for distinct $i,j \in \{1,2,3,4,5\}$, if $x \in V_{i,i+2}'$ and $y \in V_{j,j+2}'$ are adjacent in~$G''$ then they must have the same unique neighbour in~$V_\emptyset$ in the graph~$G''$.
Therefore every component of~$G''$ contains
exactly
one vertex in~$V_\emptyset$ and at most one vertex in each set~$V_{i,i+2}'$, so each component contains at most six vertices.
It follows that~$G''$ has clique-width at most~$6$.
By Fact~\ref{fact:bip}, $G'$ has bounded clique-width.
We may therefore assume that~$V_\emptyset$ is empty.
This completes the proof of Claim~\ref{clm:v0-empty}.

\medskip
\noindent
Note that in the proof of 
Claim~\ref{clm:v0-empty}
we may remove vertices from~$V_{i,i+2}$ for some $i \in \{1,2,3,4,5\}$.
If this causes a set~$V_{i,i+2}$ to become small, then again by Fact~\ref{fact:del-vert}, we may delete all remaining vertices of~$V_{i,i+2}$.
Therefore, we may again assume that each set~$V_{i,i+2}$ is either large or empty.
We now analyse the edges between these sets.

\clm{\label{clm:V13-V24-comatching}For $i \in \{1,2,3,4,5\}$, the edges between~$V_{i,i+2}$ and~$V_{i+1,i+3}$ form a co-matching.}
Suppose for contradiction that the claim is false.
Without loss of generality, assume there is a vertex~$x \in V_{1,3}$ with two non-neighbours $y,y' \in V_{2,4}$.
Then~$y$ must be non-adjacent to~$y'$ by Claim~\ref{clm:V13-indep}, so
$G[x,v_1,y,v_4,y']$ is a $P_2+\nobreak P_3$, a contradiction.
This completes the proof of Claim~\ref{clm:V13-V24-comatching}.

\clm{\label{clm:V13-adj-V24-Vjj+1-not-comp}For $i \in \{1,2,3,4,5\}$, suppose $x \in V_{i,i+2}$ is adjacent to $y \in V_{i+1,i+3}$. If $z \in V_{i+3,i}$ then~$z$ is not complete to $\{x,y\}$.}
Suppose, for contradiction that $x \in V_{1,3}$ is adjacent to $y \in V_{2,4}$ and $z \in V_{4,1}$ is complete to $\{x,y\}$.
Then $G[v_1,y,x,v_2,z]$ is a $\overline{P_2+P_3}$, a contradiction.
This completes the proof of Claim~\ref{clm:V13-adj-V24-Vjj+1-not-comp}.

\clm{\label{clm:V24-large-V13-V35-anti}For $i \in \{1,2,3,4,5\}$, if~$V_{i,i+2}$ is large then~$V_{i-1,i+1}$ is anti-complete to~$V_{i+1,i+3}$.}
Suppose the claim is false.
Without loss of generality, assume that~$V_{2,4}$ is large and $x \in V_{1,3}$ is adjacent to $z \in V_{3,5}$.
By Claim~\ref{clm:V13-V24-comatching}, the vertices~$x$ and~$z$ each have at most one non-neighbour in~$V_{2,4}$.
Since~$V_{2,4}$ is large, there must therefore be a vertex $y \in V_{2,4}$ that is adjacent to both~$x$ and~$z$.
Then $G[v_3,y,x,v_2,z]$ is a $\overline{P_2+P_3}$, a contradiction.
This completes the proof of Claim~\ref{clm:V24-large-V13-V35-anti}.

\medskip
\noindent
For distinct $i,j \in \{1,2,3,4,5\}$, we say that~$V_{i,i+2}$ and~$V_{j,j+2}$ are {\em consecutive} sets if~$v_i$ and~$v_j$ are adjacent vertices of the cycle~$C$ and {\em opposite} sets if they are not.
We consider cases depending on which sets~$V_{i,i+2}$ are large.

\thmcase{\label{case:4-Vii+2-12345-large}$V_{i,i+2}$ is large for all $i \in \{1,2,3,4,5\}$.}
By Claim~\ref{clm:V13-indep}, every set~$V_{i,i+2}$ is independent.
Then, by Claim~\ref{clm:V24-large-V13-V35-anti}, if~$V_{i,i+2}$ and~$V_{j,j+2}$ are opposite then they must be anti-complete to each other.
By Claim~\ref{clm:V13-V24-comatching}, if~$V_{i,i+2}$ and~$V_{j,j+2}$ are consecutive, then the edges between these sets form a co-matching.
By Fact~\ref{fact:del-vert}, we may delete the vertices of the cycle~$C$.
Applying a bipartite complementation between each pair of consecutive sets~$V_{i,i+2}$ and~$V_{j,j+2}$, we therefore obtain a graph of maximum degree~$2$.
By Lemma~\ref{lem:atmost-2} and Fact~\ref{fact:bip}, $G$ has bounded clique-width.
This completes Case~\ref{case:4-Vii+2-12345-large}.

\newpage
\thmcase{\label{case:4-Vii+2-1234-123-large}$V_{i,i+2}$ is large for $i \in \{1,2,3\}$ and empty for $i=5$. ($V_{i,i+2}$ may be large or empty for $i=4$.)}
By Claim~\ref{clm:V13-indep}, every set~$V_{i,i+2}$ is independent.
Then, by Claim~\ref{clm:V24-large-V13-V35-anti}, $V_{1,3}$ is anti-complete to~$V_{3,5}$ and~$V_{2,4}$ is anti-complete to~$V_{4,1}$ (note that this also holds in the case where $V_{4,1}=\emptyset$).
Therefore $V_{1,3} \cup V_{3,5}$ and $V_{2,4} \cup V_{4,1}$ are independent sets, so $G[V_{1,3} \cup V_{3,5} \cup V_{2,4} \cup V_{4,1}]$ is a bipartite $(P_2+\nobreak P_3)$-free graph, and thus has bounded clique-width by Lemma~\ref{lem:bipartite}.
By Fact~\ref{fact:del-vert}, we may delete the vertices of the cycle~$C$.
All remaining vertices are in $V_{1,3} \cup V_{3,5} \cup V_{2,4} \cup V_{1,4}$.
Therefore~$G$ has bounded clique-width.
This completes Case~\ref{case:4-Vii+2-1234-123-large}.

\thmcase{\label{case:4-Vii+2-two-large}$V_{i,i+2}$ is large for at most two values $i \in \{1,2,3,4,5\}$.}
By Claim~\ref{clm:V13-indep}, every set~$V_{i,i+2}$ is independent.
By Fact~\ref{fact:del-vert}, we may delete the vertices of the cycle~$C$.
Therefore the remainder of the graph consists of the at most two sets~$V_{i,i+2}$, and is therefore a bipartite $(P_2+\nobreak P_3)$-free graph, so it has bounded clique-width by Lemma~\ref{lem:bipartite}.
Therefore~$G$ has bounded clique-width.
This completes Case~\ref{case:4-Vii+2-two-large}.

\medskip
\noindent
By symmetry, only one case remains:

\thmcase{\label{case:4-Vii+2-124-large}$V_{i,i+2}$ is large for $i \in \{1,2,4\}$ and empty for $i \in \{3,5\}$.}
We consider two subcases:

\thmsubcase{\label{case:4-Vii+2-1234-123-large-no-V41-nbr-V13-V24}No vertex of~$V_{4,1}$ has a neighbour in both~$V_{1,3}$ and~$V_{2,4}$.}
By Claim~\ref{clm:V13-indep}, $V_{1,3}$, $V_{2,4}$ and~$V_{4,1}$ are independent.
By Fact~\ref{fact:del-vert}, we may delete the vertices of the cycle~$C$.
All remaining vertices belong to $V_{1,3} \cup V_{2,4} \cup V_{4,1}$.
Let~$V_{4,1}^*$ be the set of vertices in~$V_{4,1}$ that have a neighbour in~$V_{1,3}$ and let $V_{4,1}^{**}=V_{4,1} \setminus V_{4,1}^*$.
Note that if a vertex $x \in V_{4,1}$ has a neighbour in~$V_{2,4}$ then $x \in V_{4,1}^{**}$.
Now $V_{2,4} \cup V_{4,1}^*$ and $V_{1,3} \cup V_{4,1}^{**}$ are independent sets, so $G[V_{2,4} \cup V_{4,1}^* \cup V_{1,3} \cup V_{4,1}^{**}]$ is a bipartite $(P_2+\nobreak P_3)$-free graph and thus has bounded clique-width by Lemma~\ref{lem:bipartite}.
Therefore~$G$ has bounded clique-width.
This completes Case~\ref{case:4-Vii+2-1234-123-large-no-V41-nbr-V13-V24}.

\thmsubcase{\label{case:4-Vii+2-1234-123-large-yes-V41-nbr-V13-V24}There is a vertex $x \in V_{4,1}$ with neighbours $y \in V_{1,3}$ and $z \in V_{2,4}$.}
By Claim~\ref{clm:V13-indep}, $V_{1,3}$, $V_{2,4}$ and~$V_{4,1}$ are independent.
Then, by Claim~\ref{clm:V13-adj-V24-Vjj+1-not-comp}, $y$ must be non-adjacent to~$z$.
By Claim~\ref{clm:V13-indep}, $V_{1,3}$, $V_{2,4}$ and~$V_{4,1}$ are independent.
By Claim~\ref{clm:V13-V24-comatching}, the edges between~$V_{1,3}$ and~$V_{2,4}$ form a co-matching, so~$y$ is complete to $V_{2,4} \setminus \{z\}$ and~$z$ is complete to $V_{1,3}\setminus \{y\}$.
By Claim~\ref{clm:V13-adj-V24-Vjj+1-not-comp}, it follows that~$x$ has no neighbours in $(V_{1,3}\setminus \{y\}) \cup (V_{2,4} \setminus \{z\})$.
Then, since~$V_{1,3}$ and~$V_{2,4}$ are large, there must be adjacent vertices $y' \in V_{1,3}\setminus \{y\}$ and $z' \in V_{2,4}\setminus \{z\}$.
Now $G[x,y,z',y',z]$ is a~$C_5$, with vertices in that order; we denote this cycle by~$C'$.
By Fact~\ref{fact:del-vert}, we may delete the vertices of the original cycle~$C$.
Repeating the arguments of Claim~\ref{clm:V1cupV12-small}, but applied with respect to the cycle~$C'$ instead of~$C$, we find that at most five vertices outside~$C'$ have exactly one neighbour in~$C'$.
By Fact~\ref{fact:del-vert}, we may delete any such vertices.

Let~$y'' \in V_{1,3} \setminus \{y,y'\}$.
By Claim~\ref{clm:V13-indep}, $y''$ is non-adjacent to~$y$ and~$y'$.
Since~$z$ is complete to $V_{1,3}\setminus \{y\}$, it follows that~$y''$ is adjacent to~$z$.
By Claim~\ref{clm:V13-adj-V24-Vjj+1-not-comp} and the fact that~$z$ is adjacent to~$x$ and~$y''$, it follows that~$x$ cannot be adjacent to~$y''$.
Since~$y''$ cannot have exactly one neighbour on~$C'$, we conclude that~$y''$ must be adjacent to~$z'$.
Therefore every vertex of $V_{1,3} \setminus \{y,y'\}$ is adjacent to~$z$ and~$z'$, but no other vertices of~$C'$.
Similarly, every vertex of $V_{2,4} \setminus \{z,z'\}$ is adjacent to~$y$ and~$y'$, but no other vertices of~$C'$.

Let~$x' \in V_{4,1} \setminus \{x\}$.
By Claim~\ref{clm:V13-indep}, the set~$V_{4,1}$ is independent, so~$x'$ is non-adjacent to~$x$.
Therefore~$x'$ must have at least two neighbours in $\{y,y',z,z'\}$, but by Claim~\ref{clm:V13-adj-V24-Vjj+1-not-comp}, it cannot be complete to $\{y,z'\}$, $\{y',z'\}$ or $\{y',z\}$.
Therefore the neighbourhood of~$x'$ in~$C'$ must be $\{y,y'\}$, $\{z,z'\}$ or $\{y,z\}$.
If~$x'$ is adjacent to~$y$ then by Claim~\ref{clm:V13-adj-V24-Vjj+1-not-comp} and the fact that~$y$ is complete to $V_{2,4}\setminus \{z\}$, it follows that~$x'$ is anti-complete to $V_{2,4}\setminus \{z\}$.
Similarly, if~$x'$ is adjacent to~$z$, then it is anti-complete to $V_{1,3}\setminus \{y\}$.
This means that if~$x'$ is adjacent to~$y$ and~$z$ then it has no other neighbours, so it is a false twin of~$x$ and by Lemma~\ref{lem:false-twin} we may delete~$x'$ in this case.
Therefore~$x'$ is either adjacent to~$y$ and~$y'$, in which case~$x'$ is anti-complete to~$V_{2,4}$, or $x'$ is adjacent to~$z$ and~$z'$, in which case $x'$ is anti-complete to~$V_{1,3}$.

Let~$V_{4,1}^*$ be the set of vertices in~$V_{4,1}$ that are adjacent to~$y$ and~$y'$ and let~$V_{4,1}^{**}$ be the remaining vertices of~$V_{4,1}$.
Now $V_{2,4} \cup V_{4,1}^*$ and $V_{1,3} \cup V_{4,1}^{**}$ are independent sets.
Deleting the vertex~$x$, we obtain the graph $G[V_{2,4} \cup V_{4,1}^* \cup V_{1,3} \cup V_{4,1}^{**}]$, which is a bipartite $(P_2+\nobreak P_3)$-free graph and thus has bounded clique-width by Lemma~\ref{lem:bipartite}.
By Fact~\ref{fact:del-vert}, it follows that~$G$ has bounded clique-width.
This concludes Case~\ref{case:4-Vii+2-1234-123-large-yes-V41-nbr-V13-V24} and therefore Case~\ref{case:4-Vii+2-124-large}.

\medskip
\noindent
We have proved that~$G$ has bounded clique-width.
By Theorem~\ref{thm:colouring-for-bdd-cw}, this completes the proof of the theorem.\qedllncs
\end{proof}

\subsection{The Final Proof}\label{s-final}

We are now ready to prove Theorem~\ref{t-mainmain}.

\medskip
\noindent
\faketheorem{Theorem~\ref{t-mainmain} (restated).} {\em Let~$H$ be a graph with $H,\overline{H} \notin \{(s+\nobreak 1)P_1+\nobreak P_3, sP_1+\nobreak P_4\; |\; s\geq 2\}$.
Then {\sc Colouring} is polynomial-time solvable for $(H,\overline{H})$-free graphs if
$H$ or $\overline{H}$ is an induced subgraph of $K_{1,3}$,
$P_1+\nobreak P_4$,
$2P_1+\nobreak P_3$,
$P_2+\nobreak P_3$,
$P_5$, or $sP_1+\nobreak P_2$ for some $s\geq 0$
and it is  \NP-complete otherwise.}

\begin{proof}
We first consider the polynomial-time cases.
The {\sc Colouring} problem is solvable in polynomial time for
$(K_{1,3},\overline{K_{1,3}})$-free graphs~\cite{Ma13},
$(P_1+\nobreak P_4,\overline{P_1+P_4})$-free graphs (by Theorem~\ref{thm:colouring-for-bdd-cw} combined with the fact that they have bounded clique-width~\cite{BLM04b}),
$(2P_1+\nobreak P_3,\overline{2P_1+P_3})$-free graphs~\cite{BDJLPZ},
$(P_2+\nobreak P_3,\overline{P_2+P_3})$-free graphs (by Theorem~\ref{thm:P2P_3-co-P_2+P_3-free}),
$(P_5,\overline{P_5})$-free graphs~\cite{HL} and
$(sP_1+\nobreak P_2, \overline{sP_1+P_2})$-free graphs for $s\geq 0$~\cite{DGP14}.

Let~$H$ be a graph and suppose that {\sc Colouring} is not \NP-complete for $(H,\overline{H})$-free graphs and that neither~$H$ nor~$\overline{H}$ is isomorphic to $(s+1)P_1+\nobreak P_3$ or $sP_1+\nobreak P_4$ for $s \geq 2$.
We will show that one of the polynomial-time cases above holds.
For any $p \geq 3$, {\sc Colouring} is \NP-complete for $(C_3,\ldots,C_p)$-free graphs due to Lemma~\ref{lem:col-cycle-free}.
Therefore, we may assume without loss of generality that~$H$ is a forest.

First suppose that~$H$ contains a vertex of degree at least~$3$.
Then $K_{1,3} \ssi H$.
We may assume~$H$ contains a vertex~$x$ not in this~$K_{1,3}$, otherwise we are done.
Since~$H$ is a forest, $x$ can have at most one neighbour on the~$K_{1,3}$.
Then $\overline{2P_1+P_2} \ssi \overline{H}$ if~$x$ is adjacent to a leaf vertex of the~$K_{1,3}$ and $K_4 \ssi \overline{H}$ if it is not.
This means that the class of $(H,\overline{H})$-free graphs contains the class of $(K_{1,3},K_4,\overline{2P_1+P_2})$-free graphs for which {\sc Colouring} is \NP-complete~\cite{KKTW01}.
We may therefore assume that~$H$ does not contain a vertex of degree~$3$, so~$H$ must be a linear forest.

Now~$H$ must be $(P_1+\nobreak 2P_2,2P_3,P_6)$-free, otherwise the problem is \NP-complete on the class of $(H,\overline{H})$-free graphs by Theorem~\ref{thm:P1+2P2-2P3-P6-and-complements-free-NPC}.
Let $H_1,\ldots,H_r$ be the components of~$H$ with $|V(H_1)|\geq\cdots\geq|V(H_r)|$ for some $r\geq 1$ and note that each component~$H_i$ is isomorphic to a path.

If $r \geq 2$ then $|V(H_2)| \leq 2$ since~$H$ is $2P_3$-free.
Suppose $|V(H_2)| = 2$.
Then $r \leq 2$ and $|V(H_1)| \leq 3$, since~$H$ is $(P_1+\nobreak 2P_2)$-free.
This means that $H \ssi P_2 +\nobreak P_3$, so we are done.
We may therefore assume that all components apart from~$H_1$ are trivial, that is, $H=sP_1+\nobreak P_t$ for some $s,t \geq 0$.
Now $t \leq 5$, since~$H$ is $P_6$-free.
If $t=5$ then $s=0$ since~$H$ is $(P_1+\nobreak 2P_2)$-free, so $H=P_5$ and we are done.
If $t=4$ then $s \leq 1$ by assumption, so $H \ssi P_1+P_4$ and we are done. 
If $t=3$ then $s \leq 2$ by assumption, so $H \ssi 2P_1+P_3$ and we are done.
If $t \leq 2$ then $H \ssi sP_1+\nobreak P_2$ for some $s \geq 0$ and we are done.
This completes the proof.\qedllncs
\end{proof}

\section{The Proofs of Theorems~\ref{thm:4-col-overP8} and~\ref{thm:5-col-overP8}}\label{sec:k-col}

To prove Theorems~\ref{thm:4-col-overP8} and~\ref{thm:5-col-overP8}, we use a construction introduced by Huang~\cite[Section~2]{Hu16}. 
Recall that the chromatic number and clique number of a graph $G$ are denoted by $\chi(G)$ and $\omega(G)$, respectively.
A graph~$G$ is $k$-\emph{critical} if and only if $\chi(G) = k$ and $\chi(G-v) < k$ for every vertex~$v$ in~$G$.
A $k$-critical graph~$G$ is \emph{nice} if and only if it contains an independent set on three vertices $c_1,c_2,c_3$, such that $\omega(G) = \omega(G - \{c_1,c_2,c_3\}) = k-1$.

We describe the construction of Huang.
Let~$I$ be an instance of {\sc $3$-Sat} with variables $X = \{x_1,\dots,x_n\}$ and clauses $\mathcal{C} = \{C^1,\dots,C^m\}$.
We may assume that the variables in each clause are pairwise distinct.
Let~$H$ be a nice $k$-critical graph.
By definition~$H$ contains three pairwise non-adjacent vertices $c_1,c_2,c_3$ such that $\omega(H) = \omega(H - \{c_1,c_2,c_3\}) = k-1$.
We construct the graph~$G_{H,I}$ as follows:
\begin{itemize}
\item For each variable~$x_i$, we create a pair of literal vertices~$x_i$ and~$\overline{x_i}$ joined by an edge.
We call these vertices $X$-type.
\item For each variable~$x_i$, we create a variable vertex~$d_i$.
We call these vertices $D$-type.
\item For each clause~$C^j$, we create a subgraph~$H_j$ isomorphic to~$H$.
The three vertices of~$H_j$ corresponding to~$c_1,c_2,c_3$ are denoted $c_1^j,c_2^j,c_3^j$, respectively, and are called $C$-type; they are understood to each represent a distinct literal of~$C^j$.
The remaining vertices of~$H_j$ are called $U$-type.
\end{itemize}
We have the following additional adjacencies:
\begin{itemize}
\item Every $U$-type vertex is adjacent to every $D$-type and every $X$-type vertex.
\item Every $C$-type vertex is adjacent to the $X$-type vertex and $D$-type vertex that represent the corresponding literal and variable.
\end{itemize}

The following are~\cite[Lemma~1]{Hu16} and~\cite[Lemma~2]{Hu16}.
\begin{lemma}[\cite{Hu16}]\label{lem:huang1}
Let~$H$ be a nice $k$-critical graph.
An instance~$I$ of {\sc $3$-Sat} is satisfiable if and only if~$G_{H,I}$ is $(k+1)$-colourable.
\end{lemma}
\begin{lemma}[\cite{Hu16}]\label{lem:huang2}
Let~$H$ be a nice $k$-critical graph.
If~$H$ is $P_t$-free for any integer $t\geq 6$, then~$G_{H,I}$ is also $P_t$-free.
\end{lemma}

\medskip
\noindent
\faketheorem{Theorem~\ref{thm:4-col-overP8} (restated).}
{\em $4$-{\sc Colouring} is \NP-complete for $(P_7,\overline{P_8})$-free graphs.}

\begin{proof}
It follows from Lemma~\ref{lem:huang1} that we need only exhibit a nice $3$-critical graph~$H$ such that~$G_{H,I}$ is $(P_7,\overline{P_8})$-free for any {\sc $3$-Sat} instance~$I$.
We claim that~$C_7$ will suffice.
This is the graph used by Huang~\cite[Theorem~6]{Hu16} to show that {\sc $4$-Colouring} is \NP-complete for~$P_7$-free graphs.
He noted (and it is trivial to check) that~$C_7$ is a $P_7$-free nice $3$-critical graph and so, by Lemma~\ref{lem:huang2}, $G_{C_7,I}$ is also $P_7$-free.
It only remains to show that $G_{C_7,I}$ is $\overline{P_8}$-free.
Suppose, for contradiction, that~$G_{C_7,I}$ contains a~$\overline{P_8}$, whose vertex set we denote~$P$, as an induced subgraph.

We observe that any four vertices of~$P$ induce at least three edges and that any three vertices induce at least one edge.
So~$P$ cannot contain four vertices that are each either $X$-type or $D$-type as they would induce at most two edges.

So~$P$ contains at least five vertices that belong to copies of~$C_7$.
We recall that vertices from distinct copies of~$C_7$ are not adjacent in~$G_{C_7,I}$.
So if three of these five vertices belong to three distinct copies of~$C_7$, they induce no edge, a contradiction.
Hence there must be a copy of~$C_7$ that contains at least three of the vertices. As two of them must be non-adjacent, no other vertex of the five can belong to another copy of~$C_7$, otherwise we again have three vertices that induce no edge.
Therefore the five vertices must all belong to the same copy.
Considering the subgraphs of~$C_7$ on five vertices, we see that these five vertices induce one of $P_1+\nobreak P_4$, $P_2+\nobreak P_3$ and~$P_5$, each of which contains an independent set on three vertices.
Therefore~$P$ contains an independent set on three vertices.
This contradiction completes the proof.\qedllncs
\end{proof}

\noindent
\faketheorem{Theorem~\ref{thm:5-col-overP8} (restated).}
{\em $5$-{\sc Colouring} is \NP-complete for $(P_6,\overline{P_1+P_6})$-free graphs.}

\begin{proof}
It follows from Lemma~\ref{lem:huang1} that we need only to exhibit a nice $4$-critical graph~$H$ such that~$G_{H,I}$ is $(P_6,\overline{P_1+P_6})$-free for any {\sc $3$-Sat} instance~$I$.
We claim that the graph~$H$ of \figurename~\ref{fig:h} will suffice.
This is the graph used by Huang~\cite[Theorem~5]{Hu16} to show that {\sc $5$-Colouring} is \NP-complete for $P_6$-free graphs.
He noted, and it is easy to verify, that~$H$ is a $P_6$-free nice $4$-critical graph and so, by Lemma~\ref{lem:huang2}, $G_{H,I}$ is also $P_6$-free.
It only remains to show that~$G_{H,I}$ is $\overline{P_1+P_6}$-free.
Suppose, for contradiction, that~$G_{H,I}$ contains~$\overline{P_1+P_6}$ as an induced subgraph.
Thus~$G_{H,I}$ contains a set, denoted~$P$, of six vertices that induce a~$\overline{P_6}$ and a further vertex~$q$ adjacent to every vertex of~$P$.  

We observe that any four vertices of~$P$ induce at least three edges and that any three vertices induce at least one edge.

\begin{figure}
\begin{center}
\begin{tikzpicture}[scale=0.8]
	
	\tikzstyle{every node}=[minimum width=0.7cm]

	\node[circle,draw=black,fill=white] (c_1)  at (3, 5.2)   {$c_1$};
	\node[circle,draw=black,fill=white] (c_2)  at (0, 0)   {$c_2$};
	\node[circle,draw=black,fill=white] (c_3)  at (6, 0)   {$c_3$};
	\node[circle,draw=black,fill=white] (b)  at (3, 1.73)   {$b$};
	\node[circle,draw=black,fill=white] (e)  at (1.5, 2.6)   {$e$};
	\node[circle,draw=black,fill=white] (f)  at (4.5, 2.6)   {$f$};
	\node[circle,draw=black,fill=white] (g)  at (3, 0) {$g$};
	
	\draw (c_1) -- (b);
	\draw (c_2) -- (b);
	\draw (c_3) -- (b);
	\draw (c_1) -- (e);
	\draw (c_1) -- (f);
	\draw (c_2) -- (e);
	\draw (c_2) -- (g);
	\draw (c_3) -- (f);
	\draw (c_3) -- (g);
	\draw (e) -- (g);
	\draw (f) -- (g);
	\draw (e) -- (f);

\end{tikzpicture}
\end{center}
\caption{The graph~$H$ used in the proof of Theorem~\ref{thm:5-col-overP8}.}
\label{fig:h}
\end{figure}
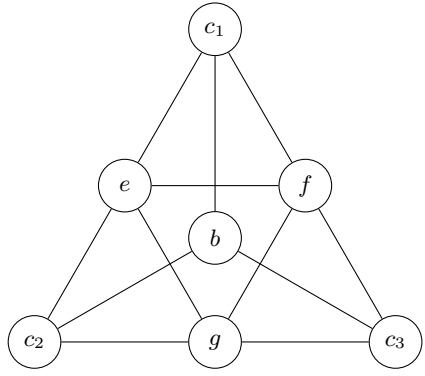

Suppose that~$q$ is either $D$-type or $X$-type.  Then~$q$ has at most one neighbour that is also $D$-type or $X$-type, and so~$P$ contains at least five vertices that belong to copies of~$H$.
Vertices from distinct copies of $H$ are not adjacent in $G_{H,I}$.
If three of these five vertices belong to three distinct copies of $H$, they induce no edge, a contradiction.  Hence the five vertices belong to at most two distinct copies of $H$.  If there are two copies of $H$ that each contain at least two of the five vertices, then there are four vertices that induce at most two edges, a contradiction.  Therefore at least four of the five vertices belong to the same copy of $H$ and so, in fact, they must all belong to the same copy (otherwise there is a single vertex in a distinct copy that is adjacent to none of the other four, a contradiction).
Label the vertices in that copy as in \figurename~\ref{fig:h}.
We notice that~$P$ contains at most one $C$-type vertex as~$q$ is adjacent to only one $C$-type vertex in each copy of~$H$
(recall that we assumed that the clauses of the {\sc $3$-Sat} instance contain three distinct variables).
So without loss of generality~$P$ contains $c_1, b, e, f$ and~$g$, but these vertices do not induce a subgraph of~$\overline{P_6}$ (notice, for example, that there are three vertices not adjacent to~$b$), a contradiction.

Let us suppose instead that~$q$ is a $C$-type vertex.
Then~$q$ has degree~$5$ in~$G_{H,I}$ so it cannot be adjacent to every vertex of~$P$, a contradiction.

Finally we assume that~$q$ is a $U$-type vertex and that it belongs to a copy of~$H$ with the labelling of \figurename~\ref{fig:h}.
 By the symmetry of~$H$, without loss of generality, only two cases remain: namely $q=b$ or $q=e$.
If $q=b$, then the possible vertices of~$P$ are $c_1, c_2, c_3$ and $D$-type and $X$-type vertices.
Hence~$P$ contains zero, one or two $C$-type vertices (it cannot contain all three since~$P$ does not contain an independent set of size~$3$).
Then~$P$ contains at least four $D$-type and $X$-type vertices, but these four vertices can induce at most two edges, a contradiction.

So we must have $q=e$ and~$P$ contains vertices from $c_1, c_2, f, g$ and the $D$-type and $X$-type vertices.
Then~$P$ cannot contain both~$c_1$ and~$c_2$ as every pair of non-adjacent vertices in a~$\overline{P_6}$ has a common neighbour (since every pair of adjacent vertices in a~$P_6$ has a common non-neighbour), but~$c_1$ and~$c_2$ have no possible common neighbour as they represent distinct variables and so are joined to different $D$-type and $X$-type vertices.
Observe that for each vertex~$v$ of~$P$, there is at least one other vertex of~$P$ that is not adjacent to~$v$.
So, looking for possible non-neighbours, we have that
\begin{itemize}
\item if~$f$ is in~$P$, then~$c_2$ is in~$P$, and
\item if~$g$ is in~$P$, then~$c_1$ is in~$P$,
\end{itemize}
and so~$P$ cannot contain both~$f$ and~$g$.
Thus~$P$ contains at least four $D$-type and $X$-type vertices, but these four vertices can induce at most two edges. This contradiction completes the proof.\qedllncs
\end{proof}

\bibliography{mybib}
\end{document}